\documentclass{article}

\usepackage[nonatbib, final]{neurips_2023}

\usepackage{shortcut}
\usepackage[utf8]{inputenc} 
\usepackage[T1]{fontenc}    
\usepackage{hyperref}       
\usepackage{url}            
\usepackage{booktabs}       
\usepackage{amsfonts}       
\usepackage{nicefrac}       
\usepackage{microtype}      
\usepackage{amsmath}
\usepackage{amssymb}
\usepackage{graphicx}
\usepackage{subcaption}
\usepackage{pifont}
\usepackage{xcolor}
\usepackage{wrapfig}
\usepackage[ruled]{algorithm2e}
\SetKwInput{KwInput}{Input}            
\SetKwInput{KwOutput}{Output} 
\usepackage{bold-extra}

\newtheorem{lemma}{Lemma}
\title{Convolutional Monge Mapping Normalization\\ for learning on sleep data}

\author{%
  Theo Gnassounou
  \\
  Université Paris-Saclay, Inria, CEA\\
  Palaiseau 91120, France \\
  \texttt{theo.gnassounou@inria.fr} \\
  \And
  R\'emi Flamary
  \\
  IP Paris, CMAP, UMR 7641\\
  Palaiseau 91120, France \\
  \texttt{remi.flamary@polytechnique.edu} \\
  \And
    Alexandre Gramfort\thanks{A. Gramfort joined Meta and can be reached at \texttt{agramfort@meta.com}}
  \\
  Université Paris-Saclay, Inria, CEA\\
  Palaiseau 91120, France \\
  \texttt{alexandre.gramfort@inria.fr}
}

\usepackage{enumitem}
\setlist{leftmargin=5.5mm}

\begin{document}

\maketitle

\begin{abstract}

  {In many machine learning applications on signals and biomedical data, especially electroencephalogram (EEG), one major challenge is the variability of the data across subjects, sessions, and hardware devices. In this work, we propose a new method called Convolutional Monge Mapping Normalization (\texttt{CMMN}), which consists in filtering the signals in order to adapt their power spectrum density (PSD) to a Wasserstein barycenter estimated on training data. \texttt{CMMN} relies on novel closed-form solutions for optimal transport mappings and barycenters and provides individual test time adaptation to new data without needing to retrain a prediction model. Numerical experiments on sleep EEG data show that \texttt{CMMN} leads to significant and consistent performance gains independent from the neural network architecture when adapting between subjects, sessions, and even datasets collected with different hardware. Notably our performance gain is on par with much more numerically intensive Domain Adaptation (DA) methods and can be used in conjunction with those for even better performances.}

\end{abstract}

\section{Introduction}

\paragraph{Data shift in biological signals}
Biological signals, such as electroencephalograms (EEG), often exhibit a significant degree of variability. This variability arises from various factors, including the recording setup (hardware specifications, number of electrodes), individual human subjects (variations in anatomies and brain activities), and the recording session itself (electrode impedance, positioning).
In this paper, we focus on the problem of sleep
staging which consists of measuring the activities 
of the brain and body during a session (here one session is done over a night of sleep)
with electroencephalograms (EEG), electrooculograms (EOG), and electromyograms (EMG) 
\cite{STEVENS200445} to classify the sleep stages. 
Depending on the dataset, the populations studied may vary from healthy
cohorts to cohorts suffering from disease \cite{MASS, goldberger2000physiobank,
SHHS}. Different electrodes positioned in the front/back \cite{goldberger2000physiobank} or the side of the head \cite{SHHS} can be employed. Sleep staging is a perfect problem for studying
the need to adapt to this variability that is also commonly denoted as data shift between domains (that can be datasets, subjects, or even sessions).



\paragraph{Normalization for data shift}

A traditional approach in machine learning to address data shifts between domains is to apply data normalization. 
Different
approaches exist in the literature to normalize data. One can normalize the
data per \texttt{Session} which allows to keep more within-session variability
in the data \cite{apicella2023effects}. If the variability during the session
is too large, one can normalize independently each window of data (\eg
30\,s on sleep EEG) \cite{chambon2017deep}, denoted as \texttt{Sample} in the following. It is also possible not to normalize
the data, and let a neural network learn to discard non-pertinent variabilities
helped with batch normalization
\cite{Supratak_2017, Perslev2021Usleep}. 
More recently, a number of works have explored the possibility of learning a
normalization layer that is domain specific in order to better adapt their
specificities \cite{li2016revisiting,chehab2022deep,liu2022convolutional,pmlr-v176-wei22a,csaky2023grouplevel,kobler2022spd,phan2023regu}.
However, this line of work usually requires to have labeled data from all domains
(subjects) for learning which might not be the case in practice when the
objective is to automatically label a new domain without an expert.


\paragraph{Domain Adaptation (DA) for data shift} Domain Adaptation is a field
in machine learning that aims at adapting a predictor in the presence of data
shift but in the absence of labeled data in the target domain. 
The goal of DA is to find an estimator using the labeled source domains which
generalizes well on the shifted target domain \cite{quinonero2008dataset}.
In the context of biological signals, DA is especially relevant as it has the ability to fine-tune the predictor for each new target domain by leveraging unlabeled data.
\\
Modern DA methods, inspired by
successes in computer vision, usually try to reduce the shift
between the embeddings of domains learned by the feature extractor. To do that, a majority of the 
methods try to minimize the divergence between the features of the source and the
target data. Several divergences can be considered for this task such as
correlation distance \cite{sun2016deep}, adversarial method
\cite{ganin2016domainadversarial}, Maximum Mean Discrepancy (MMD) distance \cite{long2015learning} or
optimal transport \cite{damodaran2018deepjdot,montesuma2021wasserstein}. Another adaptation strategy consists in learning domain-specific batch normalization \cite{li2016revisiting} in the embedding.
Interestingly those methods share
an objective with the normalization methods: they aim to reduce
the shift between datasets. To this end, DA learns a complex invariant feature representation whereas normalization remains in the original data space.
Finally, test-time DA aims at adapting a predictor to the target data
without access to the source data \cite{chen2022contrastive}, which might not be
available in practice due to privacy concerns or the memory limit of devices.


\paragraph{Contributions} In this work we propose a novel and efficient
normalization approach that can compensate at test-time for the spectral variability
of the domain signals. Our approach called Convolutional Monge Mapping Normalization
(\texttt{CMMN}), illustrated in \autoref{fig:CMMN}, uses a novel closed-form to 
estimate a meaningful barycenter from the source domains.
Then \texttt{CMMN} uses a closed-form solution for the optimal transport Monge
mapping between Gaussian random signals to align the power spectrum density of
each domain (source and target) to the barycenter. We
emphasize that \texttt{CMMN} is, to the best of our knowledge, the first approach
that can adapt to complex spectral shifts in the data without the need to access
target datasets at training time or train a new estimator for each new
domain (which are the limits of DA). \\
We first
introduce in \autoref{sec:OT} the problem of optimal transport (OT) between Gaussian
distributions, and then propose a novel closed-form solution for the Wasserstein
barycenter for stationary Gaussian random signals. We then use this result to
propose our convolutional normalization procedure in \autoref{sec:CMMN}, where implementation details and related
works are also discussed. Finally \autoref{sec:expe} reports a number of numerical experiments on sleep EEG data, demonstrating the interest of \texttt{CMMN}
for adapting to new subjects, sessions, and even datasets, but also study its
interaction with DA methods.


\begin{figure}
    \centering
    \includegraphics[width=\columnwidth]{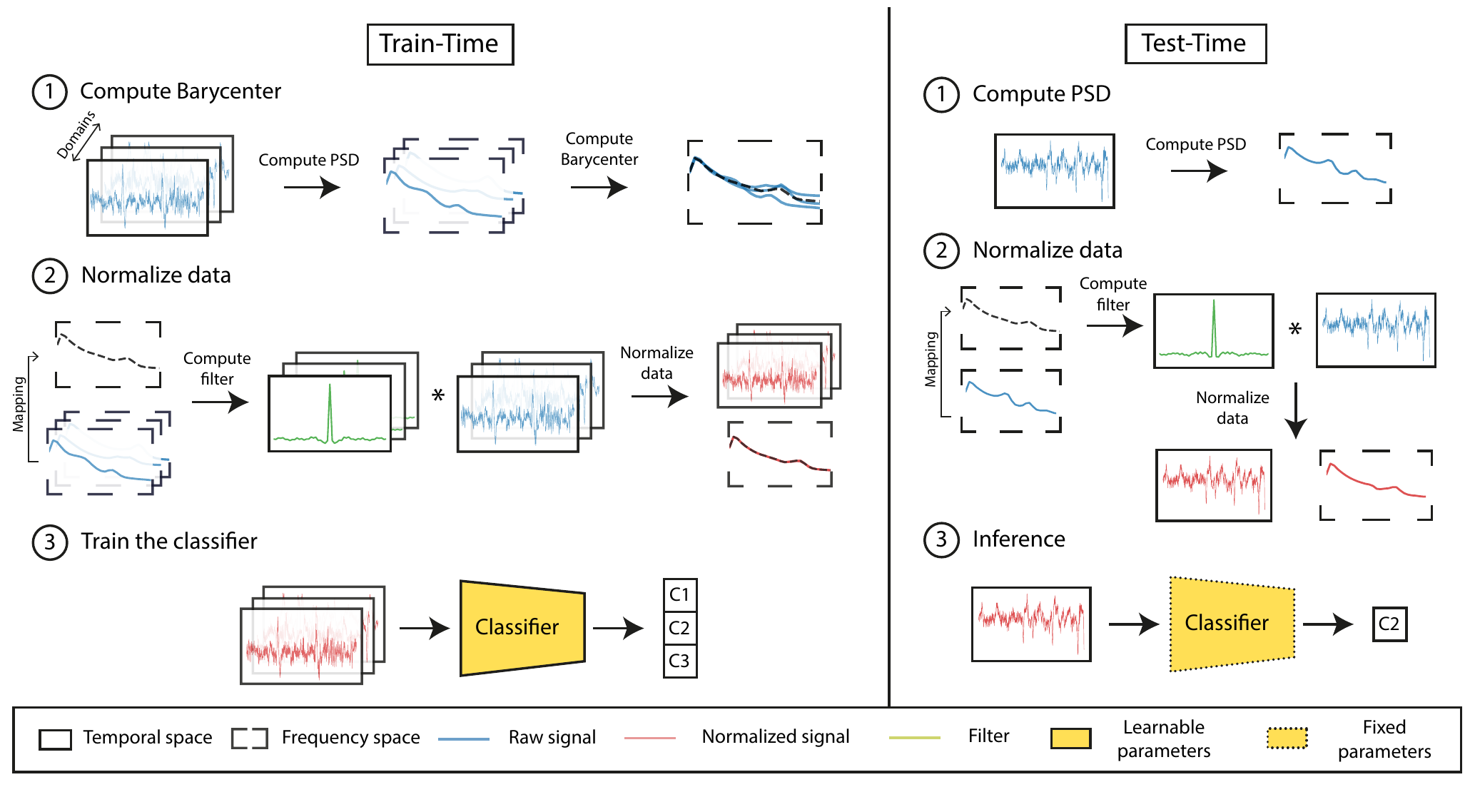}\vspace{-3mm}
    \caption{Illustration of the \texttt{CMMN} approach. At train-time the Wasserstein barycenter is estimated from 3 subjects/domains. The model is learned on normalized data. At test time the same barycenter is used to normalize test data and predict.}
    \label{fig:CMMN}
\end{figure}

\paragraph{Notations}
Vectors are denoted by small cap boldface letters (\eg $\bx$), and matrices are denoted 
by large cap boldface letters (\eg $\bX$).
The element-wise product is denoted 
by $\odot$. The element-wise power of $n$ is denoted by $\cdot^{\odot n}$.
$\intset{K}$ denotes the set $\{1, ..., K\}$. $|.|$ denotes the absolute value.
The discrete convolution operator between two signals is denoted as $*$.
Any parameters written with a small $k$ (\eg $\bX_k$ or $\bx_i^k$) is related to the source domain
$k$ with $1\leq k\leq K$. Any parameters written with a small t (\eg $\bX_t$ or $\bx_i^t$)
is related to the target domain. 

\section{Signal adaptation with Optimal Transport}
\label{sec:OT}
In this section, we first provide a short introduction to optimal transport
between Gaussian Distributions, and then discuss how those solutions can be
computed efficiently on stationary Gaussian signals, exhibiting a new
closed-form solution for Wasserstein barycenters. 

\subsection{Optimal Transport between Gaussian distributions}

\paragraph{Monge mapping for Gaussian distributions} Let two Gaussian
distributions $\mu_s = \mathcal{N}(\bm_s, \bSigma_s)$ and $\mu_t =
\mathcal{N}(\bm_t, \bSigma_t)$, where $\bSigma_s$ and $\bSigma_t$ are symmetric
positive definite covariances matrices. OT between Gaussian distributions is
one of the rare cases where there exists a closed-form solution. The OT cost,
also called the Bures-Wasserstein distance when using a quadratic ground metric, is
\cite{bhatia2019bures,peyre2020computational}
\begin{equation}
    \mathcal{W}^2_2(\mu_s, \mu_t) = \Vert \bm_s - \bm_t \Vert_2^2 + \text{Tr}\left(\bSigma_s + \bSigma_t - 2 \left(\bSigma_t^{\frac{1}{2}} \bSigma_s \bSigma_t^{\frac{1}{2}}\right)^{\frac{1}{2}}\right) \; ,
    \label{eq:bures_wass}
\end{equation}
where the second term is called the Bures metric \cite{Forrester_2015} between
positive definite matrices. The OT mapping, also called  Monge mapping, can be
expressed as the following affine function :
\begin{equation}
    m(\bx) = \bA (\bx - \bm_s) + \bm_t, \quad \text{with} \quad  \bA = \bSigma_s^{-\frac{1}{2}}\left( \bSigma_s^{\frac{1}{2}} \bSigma_t\bSigma_s^{\frac{1}{2}} \right)^{\frac{1}{2}}  \bSigma_s^{-\frac{1}{2}} = \bA^\trans\;.
    \label{eq:monge_map}
\end{equation}
In practical applications, one can estimate empirically the means and covariances of the two distributions and plug them into the equations above. Interestingly, in this case, the concentration of the estimators has been shown to be in $O(N^{-1/2})$, where $N$ is the number of samples, for the divergence estimation  \cite{nadjahi2021fast} and for the mapping estimation  \cite{flamary2020concentration}. This is particularly interesting because optimal transport in the general case is known to be very sensitive to the curse of dimensionality with usual concentrations in $O(N^{-1/D})$ where $D$ is the dimensionality of the data \cite{fournier2015rate}.

\paragraph{Wasserstein barycenter between Gaussian distributions} The Wasserstein barycenter that searches for an average distribution can also be estimated between multiple Gaussian distributions $\mu_k$. This barycenter $\bar\mu$ is expressed as
\begin{equation}
    \bar\mu  = \underset{\mu}{\text{arg min}} \; \frac{1}{K} \sum_{k=1}^{K} \mathcal{W}_2^2(\mu, \mu_k) \; .
    \label{eq:bw_bary}
\end{equation}
Interestingly, the barycenter is still a Gaussian distribution $\bar\mu =
\mathcal{N}(\bar\bm, \bar\bSigma)$ \cite{agueh2011bary}. Its mean
$\bar\bm=\frac{1}{K}\sum_k\bm_k$ can be computed as an average of the means of
the Gaussians, yet there is no closed-form for computing the covariance
$\bar\bSigma$. In practical applications, practitioners often use the following
optimality condition from \cite{agueh2011bary}
\begin{equation}
    \bar\bSigma = \frac{1}{K}\sum_{k=1}^{K} \left( \bar\bSigma^{\frac{1}{2}} \bSigma_k\bar\bSigma^{\frac{1}{2}} \right)^{\frac{1}{2}} \;,
    \label{eq:bw_fixed_point}
\end{equation}
in a  fixed point algorithm that consists in updating the covariance \cite{mroueh2019wasserstein} using equation \eqref{eq:bw_fixed_point} above until convergence. Similarly to the distance estimation and mapping estimation, statistical estimation of the barycenter from sampled distribution has been shown to have a concentration in $O(N^{-1/2})$ \cite{kroshnin2021statistical}.

 
\subsection{Optimal Transport between Gaussian stationary signals}

We now discuss the special case of OT between stationary Gaussian random signals.
In this case, the covariance matrices are Toeplitz matrices. 
A classical assumption
in signal processing is that for a long enough signal, one can assume that the signal is periodic, and therefore the covariance matrix is a
Toeplitz circulant matrix.
The circulant matrix can be diagonalized by the Discrete Fourier
Transform (DFT) $\bSigma = \bF \text{diag}(\bp) \bF^*$, with $\bF$ and $\bF ^*$
the Fourier transform operator and its inverse, and $\bp$
the Power Spectral Density (PSD) of the signal. 

\paragraph{Welch PSD estimation} As discussed above, there is a direct relation between
the correlation matrix of a signal and its PSD. In practice, one has access to a
matrix $\bX \in \mathbb{R}^{N \times T}$ containing $N$  signals of length $T$
that are samples of Gaussian random signals. In this case, the PSD $\bp$ of one
random signal can be estimated using the Welch periodogram method \cite{Welch:67} with $\hat \bp =
\frac{1}{N} \sum_{i=1}^N |\bF \bx_i|^{\odot 2}$ where $|\cdot|$ is the element-wise magnitude of the complex vector.

\paragraph{Monge mapping between two Gaussian signals} The optimal transport mapping
between Gaussian random signals can be computed from \eqref{eq:monge_map} and
simplified by using the Fourier-based eigen-factorization of the covariance matrices. The
mapping between two stationary Gaussian signals of PSD respectively $\bp_s$ and 
$\bp_t$ can be expressed with the following convolution \cite{flamary2020concentration}:
\begin{equation}
    m(\bx) = \bh * \bx\;, \quad \text{with} \quad \bh = \bF^* \left(\bp_t^{\odot\frac{1}{2}}\odot \bp_s^{\odot-\frac{1}{2}} \right)\;.
    \label{eq:monge_map_conv}
\end{equation}
Note that the bias terms $\bm_s$ and $\bm_t$ do not appear above because one can suppose in practice that the signals are centered (or have been high-pass filtered to be centered), which means that the processes are zero-mean.
The Monge mapping is a convolution with a filter $\bh$ that can be efficiently
estimated from the PSD of the two signals. It was suggested in
\cite{flamary2020concentration} as a
Domain Adaptation method to compensate for convolutional shifts between datasets. For instance, it enables compensation for variations such as differing impedances, which can be physically modeled as convolutions \cite{oppen1996elec}.
Nevertheless, this paper focused on theoretical results, and no evaluation on real signals is reported. Moreover, the method proposed in \cite{flamary2020concentration} cannot be used between
multiple domains (as explored here).
This is why in the following we propose a novel closed-form for estimating a barycenter of Gaussian signals that we will use for the normalization of our \texttt{CMMN} method.

\paragraph{Wasserstein barycenter between Gaussian signals}
As discussed above, there is no known closed-form solution for a Wasserstein barycenter between Gaussian distributions.
Nevertheless, in the case of stationary Gaussian signals, one can exploit the structure of the covariances to derive a closed-form solution that we propose below.
\begin{lemma}
Consider $K$ centered stationary Gaussian signals of PSD $\bp_k$ with $k \in \intset{K}$,
the Wasserstein barycenter of the $K$ signals is a centered stationary Gaussian signal of
PSD $\bar\bp$ with:
\begin{equation}
    \label{eq:barycenter}
    \bar\bp =\left(\frac{1}{K} \sum^{K}_{k=1} \bp_k^{\odot\frac{1}{2}}\right)^{\odot2} \; .\\
\end{equation}
\end{lemma}

\begin{proof}
The proof is a direct application of the optimality condition
\eqref{eq:bw_fixed_point} of the barycenter.  The factorized
covariances in \eqref{eq:bw_fixed_point}, the matrix square root and the inverse
can be simplified as element-wise square root and inverse, recovering equation
\eqref{eq:barycenter}. We provide a detailed proof in the appendix.
\end{proof}

The closed-form solution is notable for several reasons. First, it is a
novel closed-form solution that avoids the need for iterative fixed-point algorithms
and costly computations of matrix square roots.
Secondly, the utilization of the Wasserstein space introduces an alternative approach to the conventional $\ell_2$ averaging of Power Spectral Density (PSD). This approach involves employing the square root, like the Hellinger distance \cite{bhatia2019matrix}, potentially enhancing robustness to outliers.
Note that while other estimators for PSD averaging could be used this choice is motivated here by the fact that we use OT mappings and that the barycenter above is optimal \emph{w.r.t.} those OT mappings.

\section{Multi-source DA with Convolutional Monge Mapping Normalization}
\label{sec:CMMN}
We now introduce the core
contribution of the paper, that is an efficient method that allows to adapt to the
specificities of multiple domains and train a predictor that can generalize to
new domains at test time without the need to train a new model.
We recall here that we
have access to $K$ labeled source domains $(\bX_k,\by_k)_k$. We assume that each
domain contains $N_k$ centered signals $\bx_i^k$ of size $T$.

\paragraph{ {\texttt{CMMN}} at train time} The proposed approach, illustrated in
Figure \ref{fig:CMMN} and detailed in Algorithm \ref{alg:cmmn_train},
 consists of the
following steps: \vspace{-2mm}
{ 
\begin{enumerate}[noitemsep]\setlength{\itemsep}{1pt}
  \item Compute the PSD $\hat\bp_k$ for each source domain and use them to estimate the
  barycenter $\hat\bp$ with \eqref{eq:barycenter}.
  \item Compute the convolutional mapping $\bh_k$ \eqref{eq:monge_map_conv} between each source domain and
  the barycenter $\bar\bp$.
  \item Train a predictor on the normalized source data using the mappings $\bh_k$: \vspace{-.5mm}
  \begin{equation}
    \min_f\quad  \sum_{k=1}^{K} \sum_{i=1}^{N_k} L\left(y^{k}_i,f(\bh_k * \bx^{k}_i)\right)\;.
    \label{eq:cmm}
\end{equation} \vspace{-3mm}
\end{enumerate}
} In order to keep notations simple, we consider here the case for a single sensor, but \texttt{CMMN}
can be extended to multi-sensor data by computing independently the Monge mapping for each sensor.
Note that steps 1 and 2 can be seen as a pre-processing and are independent of
the training of the final predictor, so \texttt{CMMN} can be used as 
pre-processing on any already existing learning framework.


\begin{wrapfigure}[19]{r}{0.45\textwidth}
   \vspace{-4mm}
  \begin{minipage}{.45\textwidth}
    \begin{algorithm}[H]
     
     \DontPrintSemicolon
   \KwInput{$f$, $F$, $\{\bX_k\}_k^K$}
   \For{$k=1 \to K$}
     {   
         $\hat \bp_k \leftarrow $ Welch PSD estimation of $\bX_k$\; 
     }
     $\bar\bp$ $\leftarrow$ Compute barycenter with \eqref{eq:barycenter}\;
     \For{$k=1 \to K$}
     {   
         $\bh_k$ $\leftarrow$ Compute mapping from \eqref{eq:monge_map_conv}
     }
     $\hat f\leftarrow$ Train on adapted data with \eqref{eq:cmm}\;
     \Return $\hat f$, $\bar\bp$
     
  \caption{Train-Time \texttt{CMMN} \label{alg:cmmn_train}}
  \end{algorithm} 
  \end{minipage}
  \\
  \begin{minipage}{.465\textwidth}
    \begin{algorithm}[H]
     \DontPrintSemicolon
   \KwInput{$\hat f$, $\bar\bp$, $\bX_t$}
   $\hat \bp_t\leftarrow $  Welch PSD estimation of $\bX_t$\;
     $\bh_t$ $\leftarrow$ compute mapping from \ref{eq:monge_map_conv}\;
    \Return  $\hat \by_t =\hat f(\bh^t_k * \bX_k^t)$
  \caption{Test-Time \texttt{CMMN} \label{alg:cmmn_test}}

  \end{algorithm} 
  \end{minipage}
\end{wrapfigure}
\paragraph{\texttt{CMMN} at test time} At test time, one has access to a new unlabeled
target domain $(\bX_t)$ and the procedure, detailed in Algorithm
\ref{alg:cmmn_test}, is very simple.  One can estimate the
PSD $\hat\bp_t$ from the target domain unlabeled data and compute the mapping $\bh_t$ to the barycenter $\bar\bp$. Then the final
predictor for the new domain is $f^t(\bx_t)=f(\bh_t*\bx_t)$, that is the composition of the domain-specific
mapping to the barycenter of the training data, and the already trained predictor $f$. This is a very efficient test-time adaptation approach that only requires an estimation of the
target domain PSD that can be done with few unlabeled target samples. Yet, it allows for 
a final predictor to adapt to the spectral
specificities of new domains thanks to the convolutional Monge normalization.

\paragraph{Numerical complexity and filter size} The numerical complexity of the
method is very low as it only requires to compute the PSD of the domains and
the barycenter in $O\left(\sum_k^KN_kT\log(T)\right)$. 
It is also important to note that in practice,
the method allows for a size of normalization filter $F$ that is different
(smaller) than the size $T$ of the signals. This consists in practice in
estimating the PSD using Welch periodogram on signal windows of size $F\leq T$ that can be extracted from the raw signal or from already
extracted fixed sized source training samples.
Indeed, if we use $F=T$ then the estimated
average PSD can be perfectly adapted by the mapping, yet using many
parameters can lead to overfitting which can be limited using a smaller filter size $F\leq
T$, 
In fact, it is interesting to note that the special case $F=1$ boils down
to a scaling of the whole signal similar to what is done with a simple z-score operation. This means
that the filter size $F$ is an hyperparameter that can be tuned on the data.
From an implementation point of view, one can use the Fast Fourier Transform (FFT) to compute the convolution (for large
filters) or the direct convolution, which both have very
efficient implementation on modern hardware (CPU and GPU).

\paragraph{Related Works}

\texttt{CMMN} is a computationally efficient approach that benefits from a wide array of recent
results in optimal transport and domain adaptation. The idea of using optimal
transport to adapt distributions was first proposed in \cite{courty2014domain}.
The idea to compute a Wasserstein barycenter of distributions from multiple
domains and use it for adaptation was introduced in \cite{montesuma2021wasserstein}. Both of those approaches
have shown encouraging performances but were strongly limited by the numerical
complexity of solving the OT problems  (mapping and barycenters) on large datasets ($O(\sum_{k=1}^{K} N_k^3\log(N_k))$ or at least quadratic in $N_k$ for entropic OT). \texttt{CMMN} does not suffer from this limitation as it relies on
both Gaussian and stationary signals assumptions that allow to estimate all the
parameters for a complexity $O(\sum_{k=1}^{K} N_k\log(N_k))$ linear with the number of samples $N_k$, and quasi-linear in dimensionality of the signals $T$.
The use of Gaussian modeling and convolutional Monge
mapping for DA was first proposed in \cite{flamary2020concentration} but the
paper was mostly theoretical and only focused on the standard 2-domain DA problem
whereas \texttt{CMMN} handles multi-source and provides test-time adaptation.\\
Finally, \texttt{CMMN} also bears resemblance with the convolutional normalization layer
proposed in  \cite{liu2022convolutional} that also uses the FFT for fast
implementation. Yet,  it needs to be trained using labeled source and target data, which prevents its use on DA at test time on new unseen domains.

\section{Numerical experiments}
\label{sec:expe}
In this section, we evaluate \texttt{CMMN} on the clinical application of sleep stage
classification from EEG signals with \cite{chambon2017deep,xsleepnet}. In the following one domain can be a session of one subject (\ie in the domain-specific experiment in \autoref{sec:subjectspecific}) or one subject (\ie all other experiments). We first
compare \texttt{CMMN} to classical normalization methods and to subject-specific
normalizations. Next, we illustrate the behavior of \texttt{CMMN} when used with different
neural network architectures, and study the effect of \texttt{CMMN} on low-performing
subjects. Finally, we study the use of \texttt{CMMN} in conjunction
with domain adaptation approaches. In order to promote research reproducibility,
code is available on github \footnote{\url{https://github.com/PythonOT/convolutional-monge-mapping-normalization}}, and the datasets used are publicly available.



\subsection{Experimental setup}
\label{sec:setup}

\label{sec:dataset}

\paragraph{Sleep staging datasets} We use three publicly available datasets: Physionet (a.k.a SleepEDF)
\cite{goldberger2000physiobank}, SHHS \cite{SHHS,SHHS2} and MASS \cite{MASS}. On
all datasets, we want to perform sleep staging from 2-channels EEG signals. The
considered EEG channels are 2 bipolar channels, Fpz-Cz and Pz-Cz that have been
known to provide discriminant information for sleep staging. Note that those
channels were not available on the SHHS dataset, and we used the C3-A2 and C4-A1
instead. This brings another level of variability in the data. More details
about the datasets are available as supplementary material.

\paragraph{Pre-processing} For all experiments, we keep 60 subjects of each
dataset and the two EEG channels. The same
pre-processing is applied to all sensors. First, the recordings are
low-pass filtered with a 30\,Hz cutoff frequency, then the signals are resampled to 100\,Hz. 
Then we extract 30\,s epochs 
having a unique class.
This pre-processing is common in the field~\cite{chambon2017deep, Stephansen_2018}. 
All the data extraction and the pre-processing steps are done with MNE-BIDS
\cite{Appelhoff2019} and MNE-Python \cite{GramfortEtAl2013a}.


\paragraph{Neural network architectures and training}
\label{sec:network}
Many neural network architectures dedicated to sleep staging have been proposed \cite{Perslev2021Usleep, Supratak_2017, Eldele2021atten}.
In the following, we choose to focus on two architectures: \texttt{Chambon} \cite{chambon2017deep} and \texttt{DeepSleepNet} \cite{Supratak_2017}. For both architectures, we use the implementation from {braindecode} \cite{braindecode}.
\texttt{Chambon} is an end-to-end neural network proposed to deal with multivariate time series and is composed of two convolutional layers with non-linear activation functions.
\texttt{DeepSleepNet} is a more complex model with convolutional layers, non-linear activation functions, and a Bi-LSTM to model temporal sequences.\\
We use the Adam optimizer  with a learning rate of $10^{-3}$ for
\texttt{Chambon} and $10^{-4}$ with a weight decay of $1 \times 10^{-3}$ for
\texttt{DeepSleepNet}. The batch size is set to $128$ and the early stopping is
done on a validation set corresponding to {$20\%$ of the subjects in the training set} with 
a patience of $10$ epochs. For all methods, we optimize the cross entropy with
class weight, which amounts to optimizing for the balanced accuracy (BACC).\\
Various metrics are commonly used in the field such as Cohen's kappa, F1-score, or Balanced Accuracy (BACC) \cite{xsleepnet,chambon2017deep}. We report here the BACC score as it is a metric well adapted to unbalanced classification tasks such as sleep staging. We also report in some experiments the gain the balanced accuracy, when using \texttt{CMMN}, of the 20\% worst performing domains/subjects in the target
domain denoted as $\Delta$BACC@20 in the following.

\paragraph{Filter size and sensitivity analysis for \texttt{CMMN}}
\label{sec:sensitivity}
Our method has a unique hyperparameter that is the size of the filter $F$ used.
In our experiments, we observed that, while this
parameter has some impact, it is not critical and has quite a wide range of values ($[8,512]$) that leads to
systematic performance gains. We provide in supplementary material a sensitivity analysis of the performance for different adaptation scenarios (pairs of datasets). It shows that the
value $F=128$ is a good trade-off that we used below for all experiments.


\subsection{Comparison between different normalizations}
\label{sec:normalization}
We now evaluate the ability of \texttt{CMMN} to adapt to new subjects within
and across datasets and also between two sessions of the same subject.


\paragraph{Classical normalizations} We first compare \texttt{CMMN} to several classical normalization strategies. We compare the use of raw data \cite{Supratak_2017} letting the neural network learn the normalization from the data (\texttt{None}), standard normalization of each 30\,s samples \cite{chambon2017deep} that discard global trend along the session (\texttt{Sample}) and finally normalization by session that consists in our case to perform normalization independently on each domain (\texttt{Session}) \cite{apicella2023effects}.
We train the \texttt{Chambon} neural network on the source data of one dataset and evaluate on the target data of all other datasets for different splits.\\
\begin{table}[t]
    \begin{minipage}{0.65\linewidth}
    \centering
    \tiny
    \begin{tabular}{l|cccc}
    \toprule
    Datasets \textbackslash Norm. &                     \texttt{None} \cite{Supratak_2017}   &          \texttt{Sample} \cite{chambon2017deep}  &         \texttt{Session} \cite{apicella2023effects}&                     \texttt{CMMN} \\
    \midrule
    MASS$\rightarrow$MASS   &           $73.9 \pm 1.4$ &  $75.1 \pm 1.0$ &  $76.0 \pm 2.4$ &  $\mathbf{76.2 \pm 2.2}$ \\
    Phys.$\rightarrow$Phys. &           $68.8 \pm 2.8$ &  $69.2 \pm 2.7$ &  $69.4 \pm 3.0$ &  $\mathbf{71.7 \pm 2.4}$ \\
        SHHS$\rightarrow$SHHS   &          $55.1 \pm 12.5$ &  $61.2 \pm 3.8$ &  $60.8 \pm 2.6$ &  $\mathbf{64.3 \pm 2.7}$ \\\hline
    MASS$\rightarrow$Phys.  &           $55.9 \pm 3.1$ &  $58.4 \pm 2.4$ &  $57.5 \pm 2.0$ &  $\mathbf{62.3 \pm 1.5}$ \\
    MASS$\rightarrow$SHHS   &           $45.8 \pm 3.3$ &  $41.8 \pm 3.6$ &  $37.4 \pm 3.6$ &  $\mathbf{47.6 \pm 4.0}$ \\
    Phys.$\rightarrow$MASS  &           $63.8 \pm 3.9$ &  $64.0 \pm 2.7$ &  $63.7 \pm 2.3$ &  $\mathbf{68.3 \pm 2.5}$ \\

    Phys.$\rightarrow$SHHS  &  $\mathbf{53.9 \pm 3.2}$ &  $45.6 \pm 2.1$ &  $47.9 \pm 1.8$ &           $51.6 \pm 1.8$ \\
    SHHS$\rightarrow$MASS   &           $48.7 \pm 4.8$ &  $57.0 \pm 2.8$ &  $51.8 \pm 6.4$ &  $\mathbf{64.5 \pm 2.8}$ \\
    SHHS$\rightarrow$Phys.  &           $52.6 \pm 4.2$ &  $55.0 \pm 2.7$ &  $52.4 \pm 4.1$ &  $\mathbf{58.3 \pm 1.7}$ \\
    \hline
    Mean                    &           $57.6 \pm 4.3$ &  $58.6 \pm 2.6$ &  $57.4 \pm 3.1$ &  $\mathbf{62.7 \pm 2.4}$ \\
    \bottomrule
    \end{tabular}
    \end{minipage}
    \begin{minipage}{0.35\linewidth}
    \includegraphics[width=\linewidth]{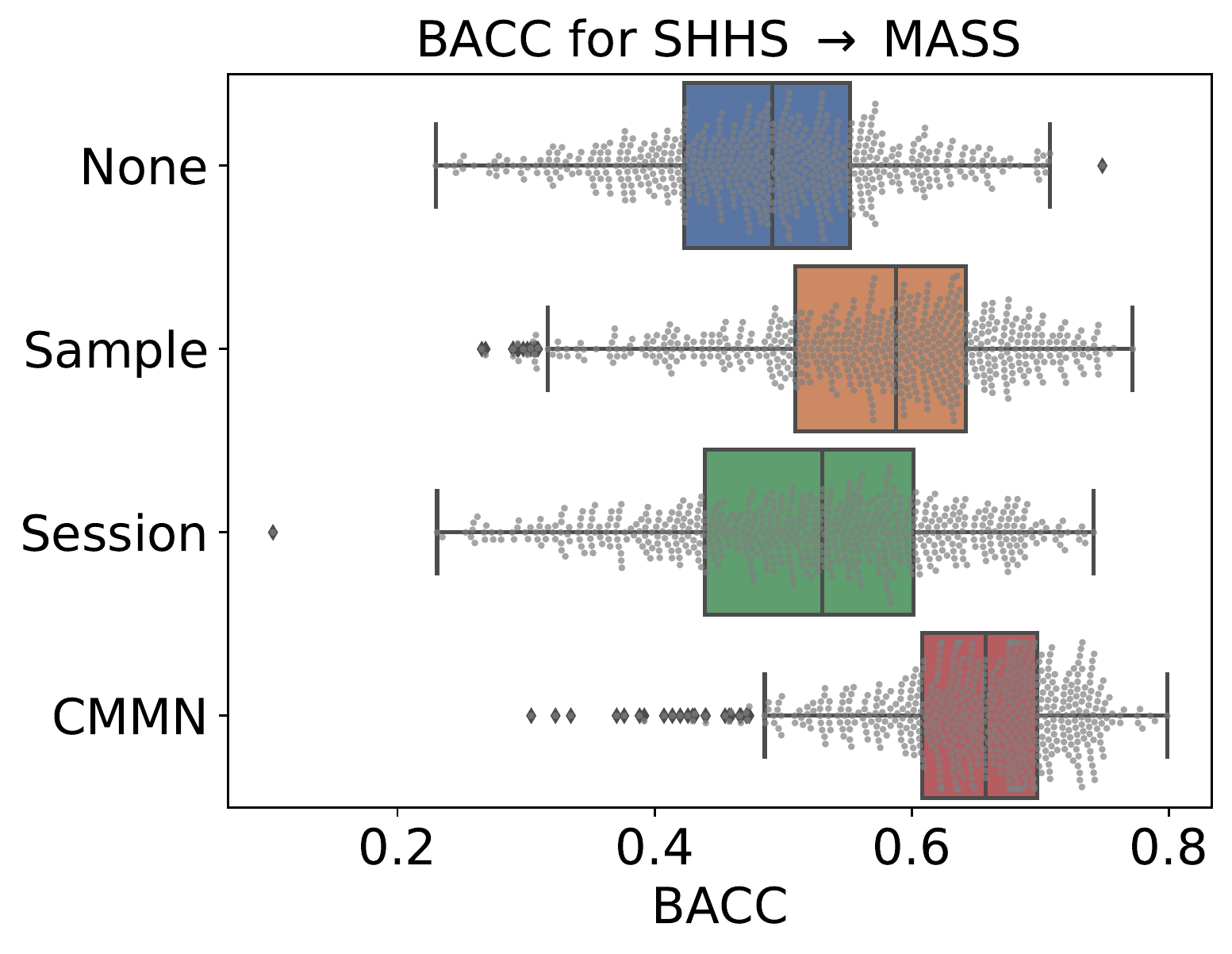}
    \end{minipage}
    \caption{Balanced accuracy (BACC) for different normalizations and different train/test dataset pairs (left). Boxplot for all normalization approaches on the specific pair SHHS$\rightarrow$MASS (right).  CMNN outperforms other normalizations.}
    \label{tab:normalization}
\end{table}The BACC for all dataset pairs and normalization are presented in the \autoref{tab:normalization}.
The three classical normalizations have similar performances with a slight edge for \texttt{Sample} on average. All those approaches are outperformed by \texttt{CMMN} in 8 out of 9 dataset pairs with an average gain of 4\% w.r.t. the best performing \texttt{Sample}.
This is also visible on the Boxplot on the right where the BACC for all subjects/domains (\ie points) is higher for \texttt{CMMN}.
Since the very simple \texttt{Sample} normalization is the best-performing competitor, we used it as a baseline called \texttt{No Adapt} in the following.



%

\paragraph{Domain specific normalization}
\label{sec:subjectspecific}
We have shown above that the \texttt{CMMN} approach allows to
better cope with distribution shifts than standard data normalizations.
This might be explained by the fact that \texttt{CMMN} normalization is domain/subject specific. This is why we now compare to several existing domain-specific normalizations.
To do this we adapt the method  of Liu \etal
\cite{liu2022convolutional} which was designed for image classification.
We implemented domain-specific convolution layers (\texttt{Conv}), 
 batch normalization (\texttt{Norm}), or both (\texttt{ConvNorm}).
In practice, we have one layer per domain that is trained (jointly with the
predictor $f$) only on data from the corresponding domain.
The limit of domain-specific normalization is that all test domains must be represented in the training set. Otherwise, if a new domain arrives in the test
set, no layer specific to that domain will have been trained.\\
To be able to compare these methods to \texttt{CMMN}, we use for 
this section the Physionet dataset for which two sessions are 
available for some subjects. The first sessions are considered 
as the training set where the domains are the subjects and the second sessions
are split between 
the validation set (20\%) and the test set (80\%). The validation 
set is used to do the early stopping, and validate the
kernel size of the subject-specific convolution for \texttt{Conv} and \texttt{ConvNorm}. \\
\begin{table}[]
    \begin{minipage}{0.5\linewidth}
    \centering
    \begin{tabular}{l|c}
    \toprule
    Normalization &      BACC           \\
    \midrule
    \texttt{No Adapt}      &  $73.7 \pm 0.7$ \\
    \texttt{Conv} \cite{liu2022convolutional}         &  $67.5 \pm 2.7$ \\
    \texttt{Norm} \cite{liu2022convolutional}         &  $69.4 \pm 1.6$ \\
    \texttt{ConvNorm} \cite{liu2022convolutional}     &  $68.1 \pm 1.3$ \\
    \texttt{CMMN}          &  $\mathbf{74.8 \pm 0.6}$ \\
    \bottomrule
    \end{tabular}
    \end{minipage}
    \hfill
    \begin{minipage}{0.5\linewidth}
    \includegraphics[width=\linewidth]{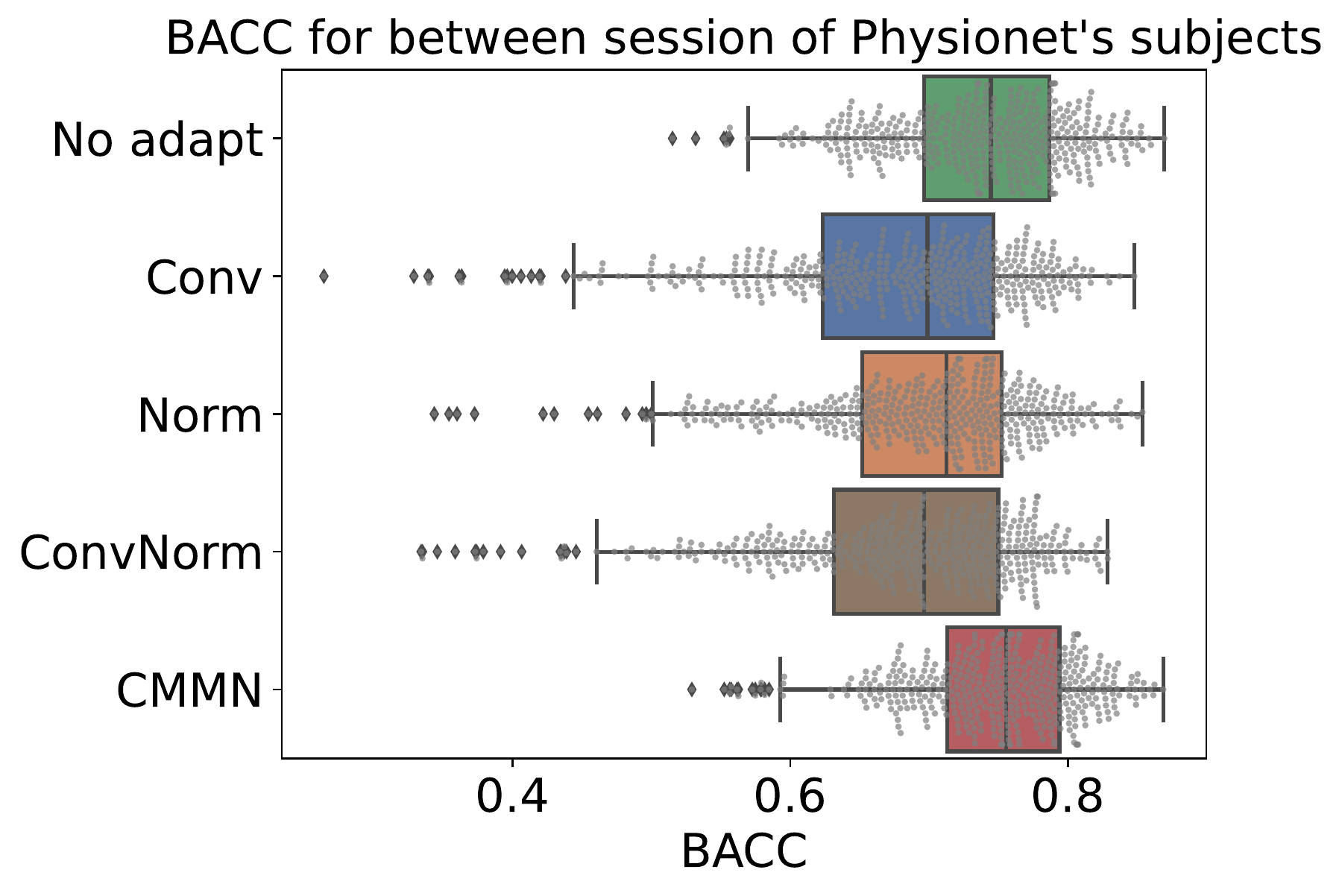}
    \end{minipage}
    \caption{Balanced accuracy (BACC) for different subject-specific normalizations and \texttt{CMMN} (left). Boxplot for all normalization approaches (right). CMNN outperforms subject-specific normalizations.}
    \label{tab:subject}
\end{table}We can see in \autoref{tab:subject} that for cross-session adaptation, the gain with \texttt{CMMN} is smaller than previous results (1\% BACC gain), which can be explained by the presence of subjects data in both domains resulting in a smaller shift between distribution.
However, \texttt{CMMN} outperforms all other subject-specific normalizations which are struggling to improve the results (\ie around 4\% of BACC loss).


\subsection{Study of the performance gain: neural architecture and human subjects}
Previous experiments have shown the superiority of \texttt{CMMN} over all the
other normalizations. In this section, we study the behavior of \texttt{CMMN} on
different neural network architectures and study which subject gains the most
performance gain.

\paragraph{Performance of \texttt{CMMN} with different architectures}
In addition to \texttt{Chambon} that was used in the previous experiments, we
now evaluate \texttt{CMMN} considering a different network architecture:
\texttt{DeepSleepNet} \cite{Supratak_2017}. The results for both architectures
are reported in \autoref{tab:archi}, where \texttt{CMMN} is consistently better for both architectures. Notably, the only configuration where the gain is limited is MASS$\rightarrow$MASS with
\texttt{DeepSleepNet} because MASS is the easiest dataset with less variability
than other pairs. Finally, we were surprised to see that \texttt{DeepSleepNet}
does not perform as well as \texttt{Chambon} on cross-dataset adaptation,
probably due to overfitting caused by a more complex architecture.


\begin{table}[]
    \small
    \centering
    
    \begin{tabular}{l|cc|cc}
    \toprule
    Architecture & \multicolumn{2}{c}{\texttt{Chambon} \cite{chambon2017deep}} & \multicolumn{2}{c}{\texttt{DeepSleepNet} \cite{Supratak_2017}} \\
     &        \texttt{No Adapt} &                     \texttt{CMMN} &        \texttt{No Adapt} &                     \texttt{CMMN} \\
    \midrule
    MASS$\rightarrow$MASS   &  $75.1 \pm 1.0$ &  $\mathbf{76.2 \pm 2.2}$ &  $\mathbf{73.3 \pm 1.7}$ &           $73.1 \pm 2.6$ \\
    Phys.$\rightarrow$Phys. &  $69.2 \pm 2.7$ &  $\mathbf{71.7 \pm 2.4}$ &  $66.5 \pm 2.5$ &  $\mathbf{69.4 \pm 2.5}$ \\
        SHHS$\rightarrow$SHHS   &  $61.2 \pm 3.8$ &  $\mathbf{64.3 \pm 2.7}$ &  $58.7 \pm 2.3$ &  $\mathbf{60.1 \pm 3.5}$ \\\hline
    
    MASS$\rightarrow$Phys.  &  $58.4 \pm 2.4$ &  $\mathbf{62.3 \pm 1.5}$ &  $50.1 \pm 2.4$ &  $\mathbf{54.5 \pm 1.2}$ \\
    MASS$\rightarrow$SHHS   &  $41.8 \pm 3.6$ &  $\mathbf{47.6 \pm 4.0}$ &  $38.3 \pm 2.6$ &  $\mathbf{47.8 \pm 2.4}$ \\
    Phys.$\rightarrow$MASS  &  $64.0 \pm 2.7$ &  $\mathbf{68.3 \pm 2.5}$ &  $59.5 \pm 1.0$ &  $\mathbf{62.1 \pm 1.9}$ \\

    Phys.$\rightarrow$SHHS  &  $45.6 \pm 2.1$ &  $\mathbf{51.6 \pm 1.8}$ &  $45.2 \pm 2.2$ &  $\mathbf{48.6 \pm 1.6}$ \\
    SHHS$\rightarrow$MASS   &  $57.0 \pm 2.8$ &  $\mathbf{64.5 \pm 2.8}$ &  $51.2 \pm 5.9$ &  $\mathbf{56.8 \pm 6.1}$ \\
    SHHS$\rightarrow$Phys.  &  $55.0 \pm 2.7$ &  $\mathbf{58.3 \pm 1.7}$ &  $48.6 \pm 5.8$ &  $\mathbf{54.7 \pm 6.8}$ \\

    \hline
    Mean                    &  $58.6 \pm 2.6$ &  $\mathbf{62.7 \pm 2.4}$ &  $54.6 \pm 2.9$ &           $\mathbf{58.6 \pm 3.2}$ \\
    \bottomrule
    \end{tabular}\vspace{1mm}
     \caption{Balanced accuracy (BACC) for different train/test dataset pairs and for different architectures (\texttt{Chambon}/\texttt{DeepSleepNet}).
    \texttt{CMMN} works independently of the network architecture.}
    \label{tab:archi}
\end{table} 

\paragraph{Performance gain on low-performing subjects} 
In medical applications, it is often more critical to have a model that has a
low failure mode, rather than the best average accuracy. 
As a first step toward studying this, we report two scatter plots reported in
\autoref{tab:delta} plotting the BACC for individual target subjects
without adaptation as a function of the BACC with \texttt{CMMN}, for different
architectures and dataset pairs.
First, the majority of the subjects are above the
axis $x=y$, which means that \texttt{CMMN} improves their score. But the most
interesting finding is the large improvement for the low-performing subjects
that can gain from 0.3 to 0.65 BACC. \\
We also provide in \autoref{tab:delta} the $\Delta$BACC@20, that is the average
BACC gain on the 20\% lowest performing subjects without adaptation.
On average, both architectures increase by 7\% the BACC on those subjects, when
it is only increased by 4\% for all subjects. Some $\Delta$BACC@20 are even greater
than 10\% on some dataset pairs. 
These results show the consistency of the method on all subjects but also the
huge impact on the more challenging ones. 


\begin{table}[]
 \begin{minipage}{0.3\linewidth}
    \centering
    \tiny
    \begin{tabular}{l|rr}
    \toprule
    Archi &     \texttt{Chambon} \cite{chambon2017deep}&    \tiny \texttt{DeepSleepNet} \cite{Supratak_2017} \\
    \midrule
    MASS$\rightarrow$MASS   &   $\mathbf{3.3 \pm 3.6}$ &   $1.8 \pm 3.8$ \\
    MASS$\rightarrow$Phys.  &   $6.4 \pm 4.1$ &   $\mathbf{9.7 \pm 2.5}$ \\
    MASS$\rightarrow$SHHS   &   $7.5 \pm 2.8$ &  $\mathbf{11.7 \pm 2.6}$ \\
    Phys.$\rightarrow$MASS  &   $\mathbf{6.9 \pm 2.8}$ &   $3.6 \pm 2.1$ \\
    Phys.$\rightarrow$Phys. &   $5.9 \pm 3.0$ &   $\mathbf{6.9 \pm 4.5}$ \\
    Phys.$\rightarrow$SHHS  &   $\mathbf{7.6 \pm 4.6}$ &   $4.5 \pm 1.7$ \\
    SHHS$\rightarrow$MASS   &  $\mathbf{15.0 \pm 5.3}$ &   $9.5 \pm 6.8$ \\
    SHHS$\rightarrow$Phys.  &   $6.5 \pm 3.4$ &   $\mathbf{9.9 \pm 6.3}$ \\
    SHHS$\rightarrow$SHHS   &   $\mathbf{4.2 \pm 3.7}$ &   $3.6 \pm 4.0$ \\
    \hline
    Mean                    &   $\mathbf{7.0 \pm 3.7}$ &   $6.8 \pm 3.8$ \\
    \bottomrule
    \end{tabular}

    \end{minipage}
    \hfill
    \begin{minipage}{0.30\linewidth}
    \includegraphics[width=\linewidth]{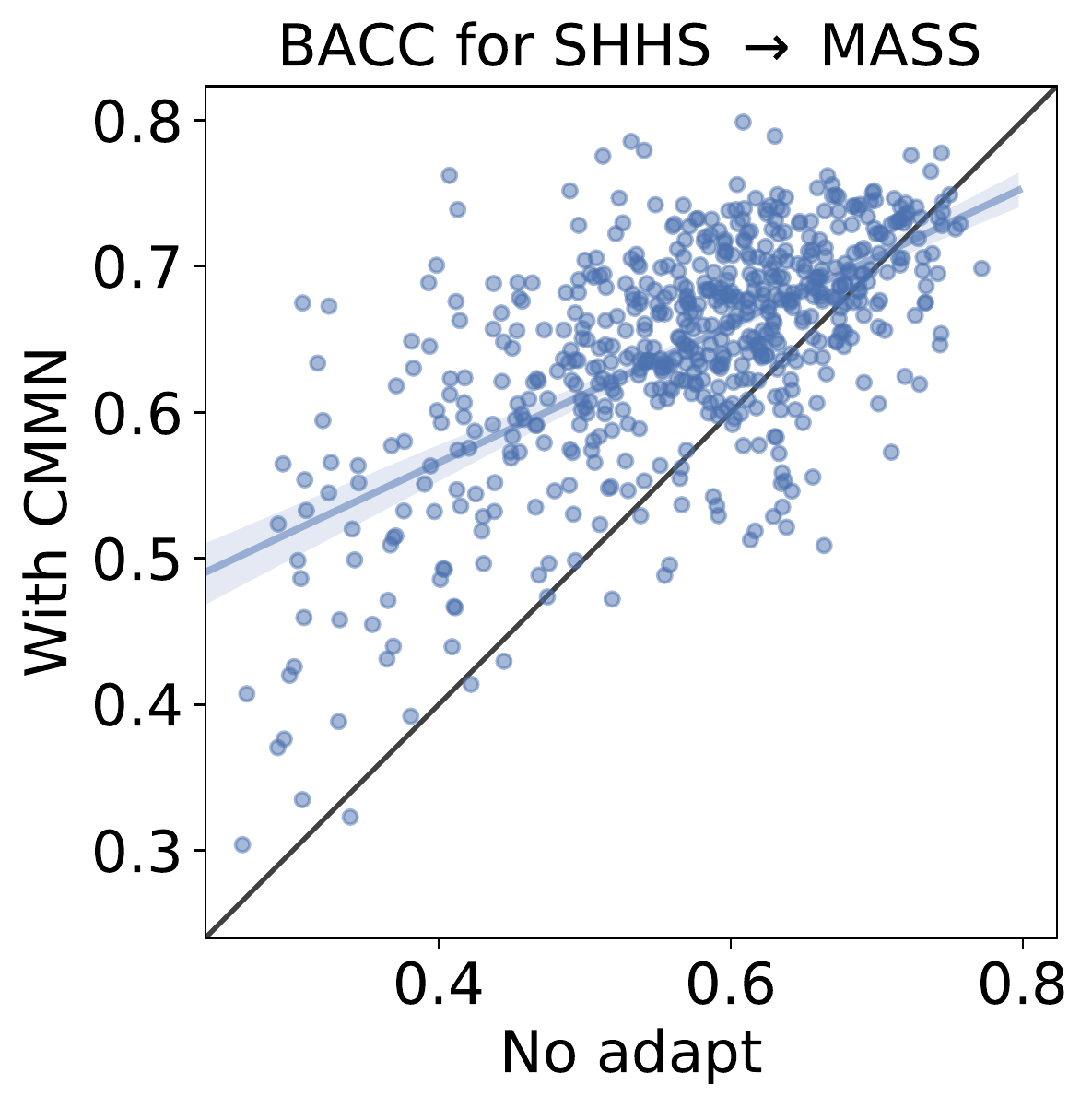}
    \end{minipage}
    \begin{minipage}{0.30\linewidth}
    \includegraphics[width=\linewidth]{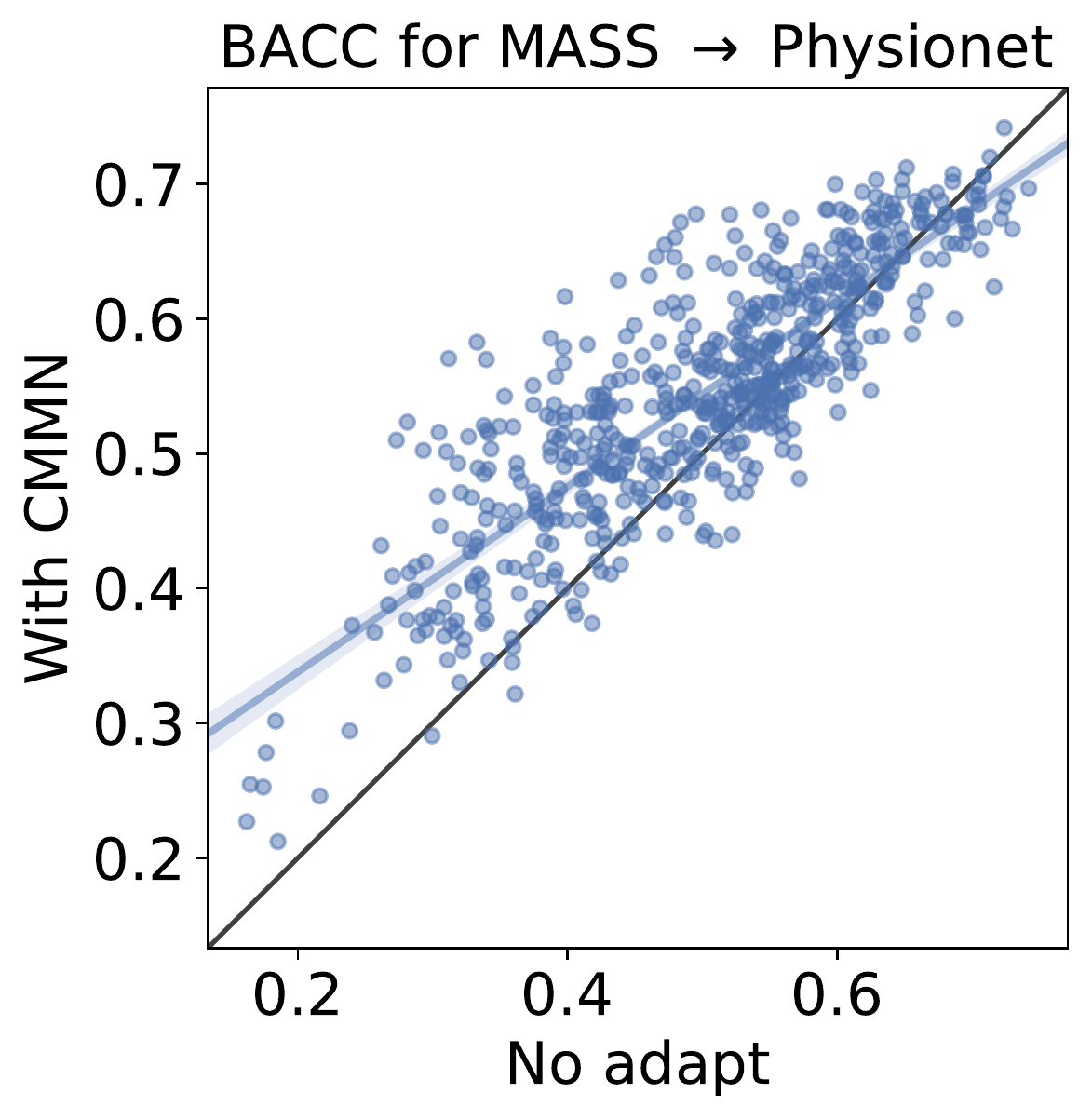}
    \end{minipage}
    \caption{$\Delta$BACC@20 for different train/test dataset pairs and for
    different architectures (\texttt{Chambon}/\texttt{DeepSleepNet}) (left).
    Scatter plot of balanced accuracy (BACC) with \texttt{No Adapt} as a fuction
    of BACC with \texttt{CMMN} and 
    the dataset pair SHHS $\rightarrow$ MASS with \texttt{Chambon} (center) and
    on the dataset pair MASS $\rightarrow$ Physionet with \texttt{DeepSleepNet}
    (right).
    \texttt{CMMN} leads to a big performance boost on the low-performing subjects.}
    \label{tab:delta}
\end{table}
\subsection{Complementarity of \texttt{CMMN} with Domain Adaptation}
\label{sec:DA}

We have shown in the previous experiments that \texttt{CMMN} is clearly the best
normalization in many settings. But the main idea of \texttt{CMMN} is to adapt
the raw signals to a common barycentric domain. Interestingly, many Domain
Adaptation (DA) methods also try to reduce the discrepancies between datasets by learning a feature representation 
 that is invariant to the domain. In this section, we compare the
two strategies and investigate if they are complementary.\\
We implement the celebrated DA method \texttt{DANN}
\cite{ganin2016domainadversarial} that aims at learning a feature representation
that is invariant to the domain using an adversarial formulation. Note that
this DA approach is much more complex than \texttt{CMMN} because it requires to
have access to the target data during training and a model needs to be trained
for each new target domain. The choice of hyperparameters for DA methods is
not trivial in the absence of target labels. But since we have access to several
target domains, we propose to select the weight parameters for \texttt{DANN}
using a validation on 20\% of the target subjects
\cite{salvador2022reproducible}. Note that this is not a realistic setting since
in real applications the target domains are usually not labeled, yet it is a way
to compare the two approaches in a configuration favorable for DA. We focus
on cross-dataset adaptation where many shifts are known to exist: different
sensors (SHHS vs Physionet/MASS), doctor scoring criteria (SHHS/MASS vs
Physionet), or brain activity (SHHS vs MASS vs Physionet).\\
We report in \autoref{tab:DA} the BACC and $\Delta$BACC@20 for all dataset pairs
and all combinations of \texttt{CMMN} and \texttt{DANN} with \texttt{Chambon}. First, we can see that
the best approaches are clearly \texttt{CMMN} and \texttt{CMMN+DANN}.
\texttt{CMMN} is better in BACC on 4/6 dataset pairs and \texttt{CMMN+DANN} is
better in $\Delta$BACC@20 on 4/6 dataset pairs. First, it is a very impressive
performance for \texttt{CMMN} that is much simpler than \texttt{DANN} and again
does not use target data when learning the predictor $f$. But it also illustrates
the interest of \texttt{CMMN+DANN} especially for
low-performing subjects.

\begin{table}[]
    \centering
    \tiny
    \begin{tabular}{l|cccc||rrr}
    \toprule
    {} & \multicolumn{4}{c}{BACC} & \multicolumn{3}{c}{$\Delta$BACC@20} \\
    Adapt &           \texttt{No Adapt} &            \texttt{DANN} &                    \texttt{CMMN} &                     \texttt{CMMN}+\texttt{DANN} &           \texttt{DANN} &           \texttt{CMMN} &            \texttt{CMMN}+\texttt{DANN} \\
    \midrule
    MASS->Phys.    &  $59.2 \pm 6.1$ &  $60.9 \pm 1.0$ &           $62.9 \pm 0.8$ &  $\mathbf{62.9 \pm 1.2}$ &  $2.0 \pm 5.9$ &   $5.3 \pm 5.8$ &   $\mathbf{5.5 \pm 6.2}$ \\
    MASS->SHHS     &  $44.5 \pm 6.0$ &  $44.3 \pm 2.1$ &  $\mathbf{50.3 \pm 4.2}$ &           $49.9 \pm 1.9$ &  $4.2 \pm 7.1$ &   $7.6 \pm 2.9$ &  $\mathbf{10.4 \pm 5.9}$ \\
    Phys.->MASS    &  $65.3 \pm 1.4$ &  $65.2 \pm 0.7$ &           $68.9 \pm 1.0$ &  $\mathbf{69.1 \pm 1.0}$ &  $0.4 \pm 1.3$ &   $\mathbf{6.3 \pm 2.0}$ &   $5.6 \pm 2.7$ \\
    Phys.->SHHS    &  $43.1 \pm 6.3$ &  $45.4 \pm 2.6$ &  $\mathbf{50.0 \pm 4.3}$ &           $49.9 \pm 2.4$ &  $2.8 \pm 5.3$ &   $8.2 \pm 5.7$ &   $\mathbf{9.7 \pm 6.1}$ \\
    SHHS->MASS     &  $59.9 \pm 3.4$ &  $59.4 \pm 1.1$ &  $\mathbf{66.5 \pm 2.5}$ &           $66.3 \pm 0.9$ &  $0.4 \pm 2.8$ &  $\mathbf{13.4 \pm 3.0}$ &  $12.6 \pm 2.7$ \\
    SHHS->Phys.    &  $57.1 \pm 3.9$ &  $57.2 \pm 2.2$ &  $\mathbf{61.6 \pm 3.1}$ &           $59.4 \pm 2.6$ &  $4.1 \pm 3.7$ &  $10.0 \pm 6.9$ &  $\mathbf{10.9 \pm 5.7}$ \\\hline
    Mean           &  $54.9 \pm 4.5$ &  $55.4 \pm 1.6$ &  $\mathbf{60.0 \pm 2.7}$ &           $59.6 \pm 1.7$ &  $2.3 \pm 4.4$ &   $8.5 \pm 4.4$ &   $\mathbf{9.1 \pm 4.9}$ \\
    \bottomrule
    \end{tabular} \vspace{1mm}
    \caption{Balanced accuracy (BACC) and $\Delta$BACC@20 for different
    train/test dataset pairs and for different adaptation methods. \texttt{CMMN}
    outperforms \texttt{DANN} and \texttt{CMMN+DANN} on average on all subjects. Combining
    \texttt{CMMN} \texttt{DANN} improves the lower-performing subjects.}
    \label{tab:DA}
   
\end{table}

\subsection{Different PSD targets}
CMMN has shown a significant boost in performance for different adaptations in various settings. 
The method consists of learning a barycenter, and mapping all domains to this barycenter to obtain a homogeneous frequency spectrum. 
Mapping to the barycenter is one way to achieve this, but it is reasonable to evaluate the optimality of this choice by comparing to alternatives. For example, the classical method called spectral whitening is equivalent to OT mapping towards a uniform PSD, which is equivalent to a white noise PSD. 
In this part, we propose to compare the mapping to different PSD targets: barycenter (classic \texttt{CMMN}), white noise PSD (whitening), or a non-uniform power-law. 
The last PSD target is a mathematical distribution that describes a functional relationship between frequency and magnitude, where magnitude is inversely proportional to the power of the frequency $P(f) = af^{a-1}$. Below we selected $a=0.659$ for the experiment. \\
The table \autoref{tab:target} gives the BACC score of the mapping for the different PSD targets and shows the importance of the reference PSD. First, we can see that mapping to a PSD increases the score significantly w.r.t. to raw data. Second, the mapping to the Wasserstein barycenter (\ie \texttt{CMMN}) is not always the better performer (only 5/9), but overall, \texttt{CMMN} gives better results (2\% higher) and with less variance. 
The robustness of the chosen PSD target, coupled with the fast computation of the barycenter, makes \texttt{CMMN} a strong and easy-to-compute normalization method for time series.
\begin{table}
\centering
\caption{BACC for different PSD targets for the mapping.}\small
\begin{tabular}{l|cccc}
\toprule
Target PSD &            None &               Barycenter &                 Powerlaw &                Whitening \\
\midrule
MASS$\rightarrow$MASS           &  $75.1 \pm 1.0$ &  $\mathbf{77.1 \pm 1.3}$ &           $75.6 \pm 2.4$ &           $73.2 \pm 6.4$ \\
MASS$\rightarrow$Physionet      &  $58.4 \pm 2.4$ &           $62.6 \pm 2.0$ &           $63.1 \pm 1.5$ &  $\mathbf{63.2 \pm 1.0}$ \\
MASS$\rightarrow$SHHS           &  $41.8 \pm 3.6$ &           $50.4 \pm 8.0$ &           $51.9 \pm 3.0$ &  $\mathbf{52.4 \pm 2.4}$ \\
Physionet$\rightarrow$MASS      &  $64.0 \pm 2.7$ &  $\mathbf{68.0 \pm 1.5}$ &           $66.7 \pm 2.5$ &           $65.9 \pm 2.6$ \\
Physionet$\rightarrow$Physionet &  $69.2 \pm 2.7$ &  $\mathbf{72.3 \pm 1.8}$ &          $66.9 \pm 16.5$ &           $71.3 \pm 1.9$ \\
Physionet$\rightarrow$SHHS      &  $45.6 \pm 2.1$ &           $52.5 \pm 1.9$ &           $53.9 \pm 1.7$ &  $\mathbf{54.0 \pm 2.4}$ \\
SHHS$\rightarrow$MASS           &  $57.0 \pm 2.8$ &  $\mathbf{65.9 \pm 3.2}$ &          $59.0 \pm 13.9$ &          $59.0 \pm 13.8$ \\
SHHS$\rightarrow$Physionet      &  $55.0 \pm 2.7$ &  $\mathbf{62.0 \pm 3.2}$ &          $56.5 \pm 13.1$ &          $55.9 \pm 12.8$ \\
SHHS$\rightarrow$SHHS           &  $61.2 \pm 3.8$ &           $63.5 \pm 3.2$ &  $\mathbf{63.6 \pm 2.6}$ &           $62.4 \pm 2.8$ \\ \hline
Mean                 &  $58.6 \pm 2.6$ &  $\mathbf{63.8 \pm 2.9}$ &           $61.9 \pm 6.4$ &           $61.9 \pm 5.1$ \\
\bottomrule
\end{tabular}
\label{tab:target}
\end{table}
\section{Conclusion}

We proposed in this paper a novel approach for the normalization of bio-signals that
can adapt to the spectral specificities of each domain while being a test-time
adaptation method that does not require retraining a new model. The method builds
on a new closed-form solution for computing Wasserstein barycenters on stationary
Gaussian random signals. We showed that this method leads to a systematic
performance gain
on different configurations of data shift (between subjects, between sessions, and
between datasets) and on different architectures. We also show that
\texttt{CMMN} benefits greatly the subjects that had bad performances when
trained jointly without sacrificing performance on the well-predicted subjects.
Finally, we show that \texttt{CMMN} even outperforms DA methods
and can be used in conjunction with DA for even better results.\\
Future work will investigate the use of \texttt{CMMN} for other biomedical
applications and study the use of the estimated filters $\bh_k$ as vector
representations of the subjects that can be used for interpretability. Finally,
we believe that a research direction worth investigating is the federated
estimation of \texttt{CMMN} with the objective of learning an unbiased estimator in
the context of differential privacy \cite{wei2020federated,kairouz2021advances}.

\newpage
\section{Acknowledgement}
The authors thank Antoine Collas, Cédric Allain, and Nicolas Courty for their valuable comments on the manuscript, and Lina Dalibard for her help with Figure 1. Numerical computation was enabled by the scientific Python ecosystem: NumPy~\cite{harris_array_2020}, SciPy~\cite{virtanen_scipy_2020}, Matplotlib~\cite{hunter_matplotlib_2007}, Seaborn~\cite{waskom_seaborn_2021}, PyTorch~\cite{paszke_pytorch_2019}, and MNE for EEG data processing \cite{GramfortEtAl2013a}. This work was partly supported by the grants ANR-20-CHIA-0016 and ANR-20-IADJ-0002 and ANR-23-ERCC-0006-01 from Agence nationale de la recherche (ANR).\\

\bibliographystyle{utphys}
\bibliography{biblio}

\newpage
\appendix
\section{Supplementary material}
\subsection{Proof of the convolutional Wasserstein barycenter}
\begin{proof}
Consider $K$ centered stationary Gaussian signals of covariance $\bSigma_k$  $\bp_k$ (respectively PSD $\bp_k$) with $k \in \intset{K}$,
the Wasserstein barycenter of the $K$ signals is a centered stationary Gaussian signal of
PSD $\bar\bp$ with:
\begin{equation}
        \bar\bSigma = \frac{1}{K}\sum_{k=1}^{K} \left( \bar\bSigma^{\frac{1}{2}} \bSigma_k\bar\bSigma^{\frac{1}{2}} \right)^{\frac{1}{2}}
\end{equation}

The signals are supposed to be stationary. Therefore the covariance matrix is a
Toeplitz circulant matrix.
The circulant matrix can be diagonalized by the Discrete Fourier
Transform (DFT) $\bSigma = \bF \text{diag}(\bp) \bF^*$, with $\bF$ and $\bF ^*$
the Fourier transform operator and its inverse, and $\bp$
the Power Spectral Density (PSD) of the signal.
The above equation becomes:
\begin{equation}
    \bar\bp = \frac{1}{K}\sum_{k=1}^{K} \left( \bar\bp^{\odot\frac{1}{2}}\odot \bp_k \odot \bar\bp^{\odot\frac{1}{2}} \right)^{\odot\frac{1}{2}}
\end{equation}
 The matrix square root and the inverse become element-wise square root and inverse. The equation becomes easier, and the term can be managed to isolate $\bar\bp$: 
\begin{equation*}
    \begin{split}
        \bar\bp &= \frac{1}{K}\sum_{k=1}^{K} \left( \bar\bp^{\odot\frac{1}{2}}\odot \bp_k \odot \bar\bp^{\odot\frac{1}{2}} \right)^{\odot\frac{1}{2}} \\
        \bar\bp &= \frac{1}{K}\sum_{k=1}^{K} \bar\bp^{\odot\frac{1}{2}}\odot \bp_k^{\odot\frac{1}{2}} \\
        \bar\bp^{\odot\frac{1}{2}} &= \frac{1}{K}\sum_{k=1}^{K} \bp_k^{\odot\frac{1}{2}} \\     
        \bar\bp &=\left(\frac{1}{K} \sum^{K}_{k=1} \bp_k^{\odot\frac{1}{2}}\right)^{\odot2} \\
    \end{split}
\end{equation*}
\end{proof}

\begin{proof}
A second possible proof is considering the optimization problem \emph{w.r.t} $\bSigma$:
\begin{equation}
    \bar\bSigma = \underset{\bSigma}{\text{arg min}} \sum_{k=1}^K \text{Tr}\left(\bSigma + \bSigma_k - 2 \left(\bSigma^{\frac{1}{2}} \bSigma_k \bSigma^{\frac{1}{2}}\right)^{\frac{1}{2}}\right) \; .
\end{equation}
As mentioned above, it is possible to use the PSD $\bp$ to transform the equation into an element-wise problem as before:
\begin{equation}
    \begin{split}
        \bar\bp &= \underset{\bp}{\text{arg min}} \sum_{k=1}^K \Vert\bp + \bp_k - 2 \left(\bp^{\frac{1}{2}} \odot \bp_k \odot \bp^{\odot\frac{1}{2}}\right)^{\odot \frac{1}{2}}\Vert_1  \\
        \bar\bp &= \underset{\bp}{\text{arg min}} \sum_{k=1}^K \Vert\bp + \bp_k - 2 \bp^{\odot \frac{1}{2}} \odot \bp_k^{\odot \frac{1}{2}}\Vert_1 \\
        \label{eq:barycenter}
    \end{split}
\end{equation}
After derivation, the $\bar\bp$ minimizing the optimization problem is given by:
\begin{equation}
    \bar\bp =\left(\frac{1}{K} \sum^{K}_{k=1} \bp_k^{\odot\frac{1}{2}}\right)^{\odot2} 
\end{equation}

\end{proof}

\subsection{Computation}
The training is done on Tesla V100-DGXS-32GB with Pytorch. We are considering the following train settings: \texttt{Chambon} architecture with a learning rate of $1e^{-3}$ for Adam optimizer and a patience of 10 for the early stopping. The training for one dataset pair with ten different splits and seeds lasts approximately 1 hour. The data processing time with the CMMN is insignificant compared to the network's computation time (a few minutes).

\subsection{Dataset descriptions}

\paragraph{SHHS}The Sleep Heart Health Study is a multi-center cohort study 
proposed by the National Heart Lung \& Blood Institute \cite{SHHS, SHHS2} to 
help detect cardiovascular disease and sleep disorders. This large dataset 
comprises 6441 subjects (age $63.1 \pm 11.2$) from 1995 to 1998. 
Five sensors are available for each subject: 2 EEGs from C3-A2 and C4-A1 channels, 
left and right EOGs, and one EMG. The EEGs have a sampling rate of 125 Hz. 
The hypnograms were scored according to the Rechtschaffen and Kales criteria \cite{Hobson1969AMO}.

\paragraph{MASS} The Montreal Archive of Sleep Studies comprised five different subsets of recordings. This paper focuses on the SS3 with recordings from 62 healthy subjects (age $42.5 \pm 18.9$). For each subject, 20 EEGs, left and right EOGs, and 3 EMGs are available. We reduced the number of EEG channels to 2 bipolar channels, Fpz-Cz and Pz-Cz, obtained by montage reformatting. The EEGs have a sampling rate of 256 Hz. The MASS hypnograms were scored according to the AASM criteria \cite{AASM}. 

\paragraph{Physionet SleepEDF} This dataset comprises two subsets, one for the age effect in healthy subjects (SC) and one for the Temazepam effect on sleep (ST). We focused on the SC subset where 78 subjects are available (age $28.7 \pm 2.9$). Each recording comprises 2 EEGs from Fpz-Cz and Pz-Cz channels, 1 EOG, and 1 EMG. The EEGs have a sampling rate of 100 Hz. Some of the subjects have two sessions of PSGs available. The hypnograms were scored according to the Rechtschaffen and Kales criteria \cite{Hobson1969AMO}. The stages N3 and N4 have been merged for the following.

\paragraph{Ethical consideration} All datasets used in our experiments are anonymized and public datasets that have already passed an ethics committee before recording.

\subsection{Sensitivity analysis to the filter size}
Several adaptations across different dataset pairs are done to compare the effect of the filter size. The smallest filter size means no transformation, and the largest size corresponds to a perfect transformation between the two signals. For each parameter, ten training are done over data from the source dataset and then tested over data from the target dataset.

\begin{figure}[t]
  \centering
  \includegraphics[scale=0.6]{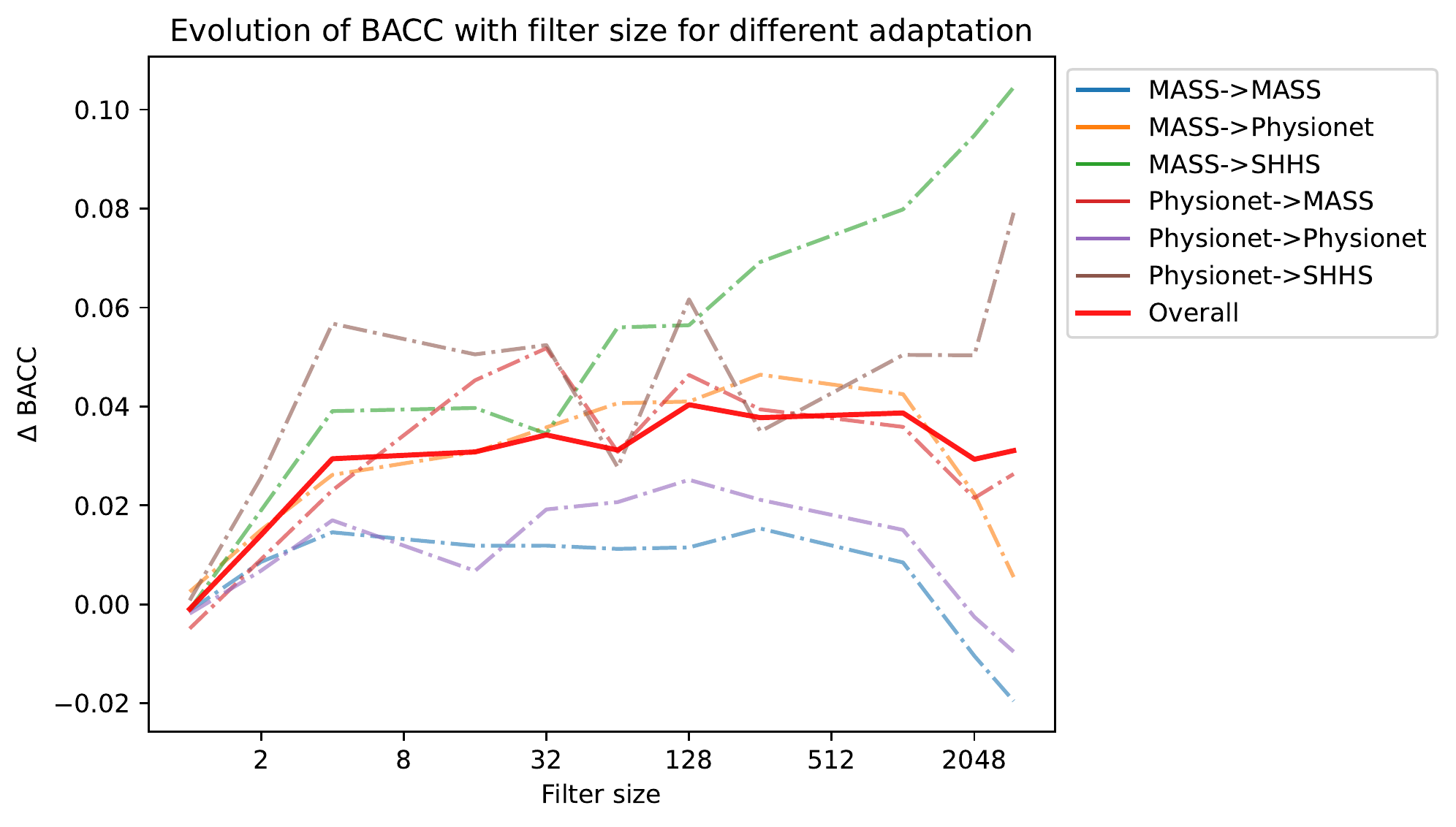}
  \caption{Evolution of the $\Delta$BACC (BACC - BACC without mapping) for different filter sizes for different adaptation problems}
  \label{fig:sensitivity}
\end{figure}

To evaluate the benefit of the method, we measure the $\Delta$BACC corresponding to the difference between the balanced accuracy score with monge mapping and the balanced accuracy score without monge mapping (\ie using \texttt{Sample} normalization). The \autoref{fig:sensitivity} shows the evaluation of $\Delta$BACC with filter size for different dataset pairs. The slightest improvement is for adaptation between the same dataset, which is logical because there is less difference to compensate between the subjects. And the best improvement is for the most challenging task, adaptation between datasets with different sensors (MASS/Physionet $\rightarrow$ SHHS). The mapping did not capture enough information for the smallest filter size to reduce the difference between distributions. On the other hand, the bigger the distribution gap between datasets, the bigger the filter size helps to adapt. Indeed, for an adaptation between the same dataset, having a filter size close to the sample size decreases the performance (see MASS  $\rightarrow$ MASS, Physionet  $\rightarrow$ Physionet), while for an adaptation between two different or very different datasets, increasing the filter size causes the performance to remain the same (MASS $\leftrightarrow$ Physionet) or even increases the scores considerably (MASS/Physionet $\rightarrow$ SHHS).

\subsection{Boxplot of BACC for different data normalizations, different architectures, and different dataset pairs}
\label{sec:}
As shown in the experimental section of the paper, CMMN outperforms standard normalization (\texttt{Session} or \texttt{Sample}) for different architectures. Here we provide more boxplots for other dataset pairs than in the main article. For \texttt{Chambon}, the figure \ref{fig:boxplot_chambon} shows again that \texttt{CMMN} outperforms other normalizations except with dataset SHHS in  target. Indeed, if \texttt{CMMN} is still better than \texttt{Sample} and \texttt{Session}, using no normalization is better for this experimental setting. SHHS is the most different dataset since the sensors used differ from Physionet and MASS, which can explain this difference. Even when SHHS is in train, \texttt{CMMN} is better than \texttt{None}. 

The results for \texttt{DeepSleepNet} in figure \ref{fig:boxplot_deep} are slightly different. If \texttt{CMMN} is still the better performer overall, the best standard normalization is \texttt{None}. Using no normalization is better in 5/6 dataset pairs over \texttt{Sample} and \texttt{Session}. These results were expected since no normalization was used in the paper proposing \texttt{DeepSleepNet} \cite{Supratak_2017}. 

For the sake of simplicity, we used \texttt{Sample} normalization before \texttt{CMMN} and also for the baseline \texttt{No Adapt.} even for \texttt{DeepSleepNet}.
\begin{figure}
\centering
\begin{subfigure}{.5\textwidth}
  \centering
  \includegraphics[width=.9\linewidth]{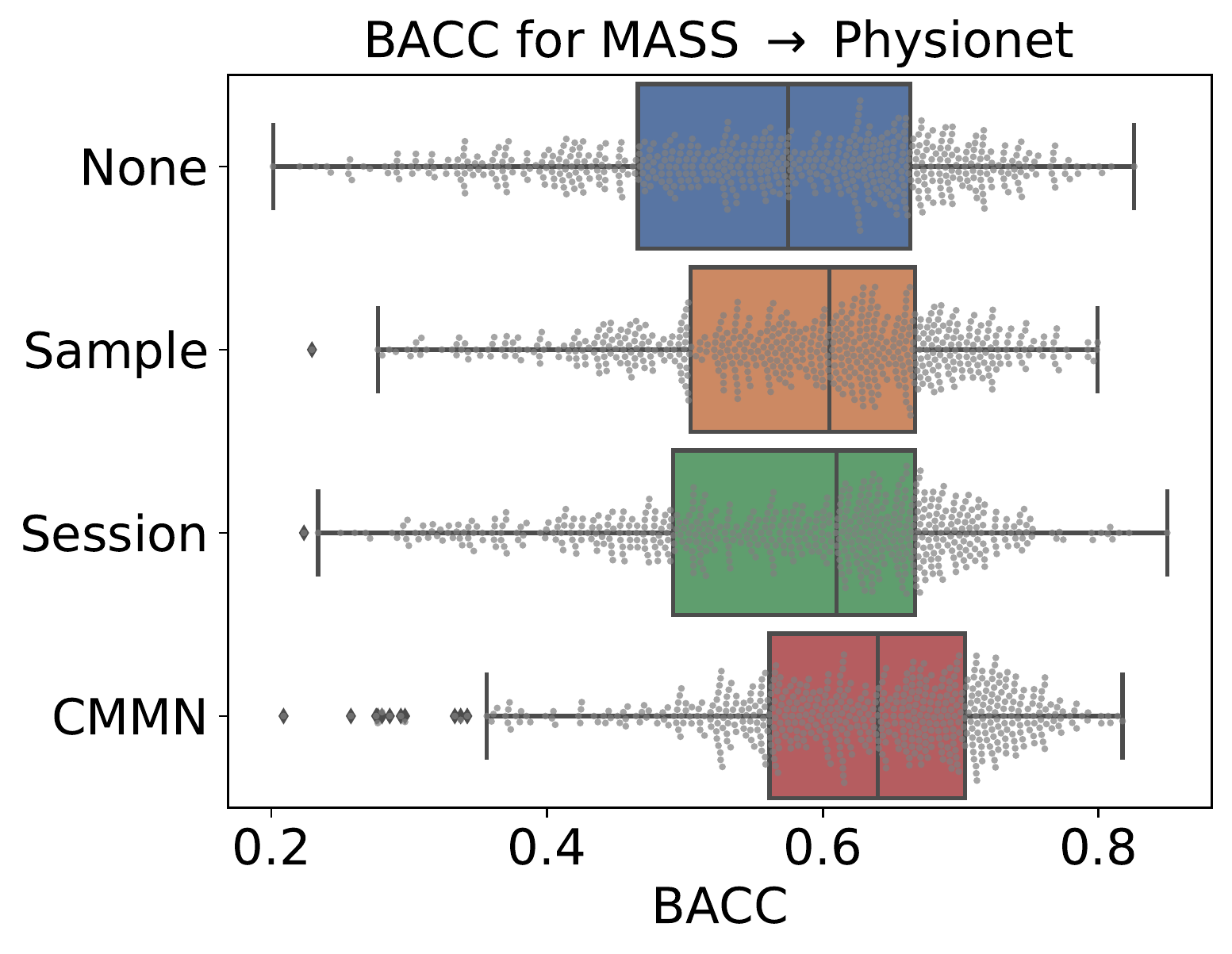}
  \caption{}
  \label{fig:sub1}
\end{subfigure}%
\begin{subfigure}{.5\textwidth}
  \centering
  \includegraphics[width=.9\linewidth]{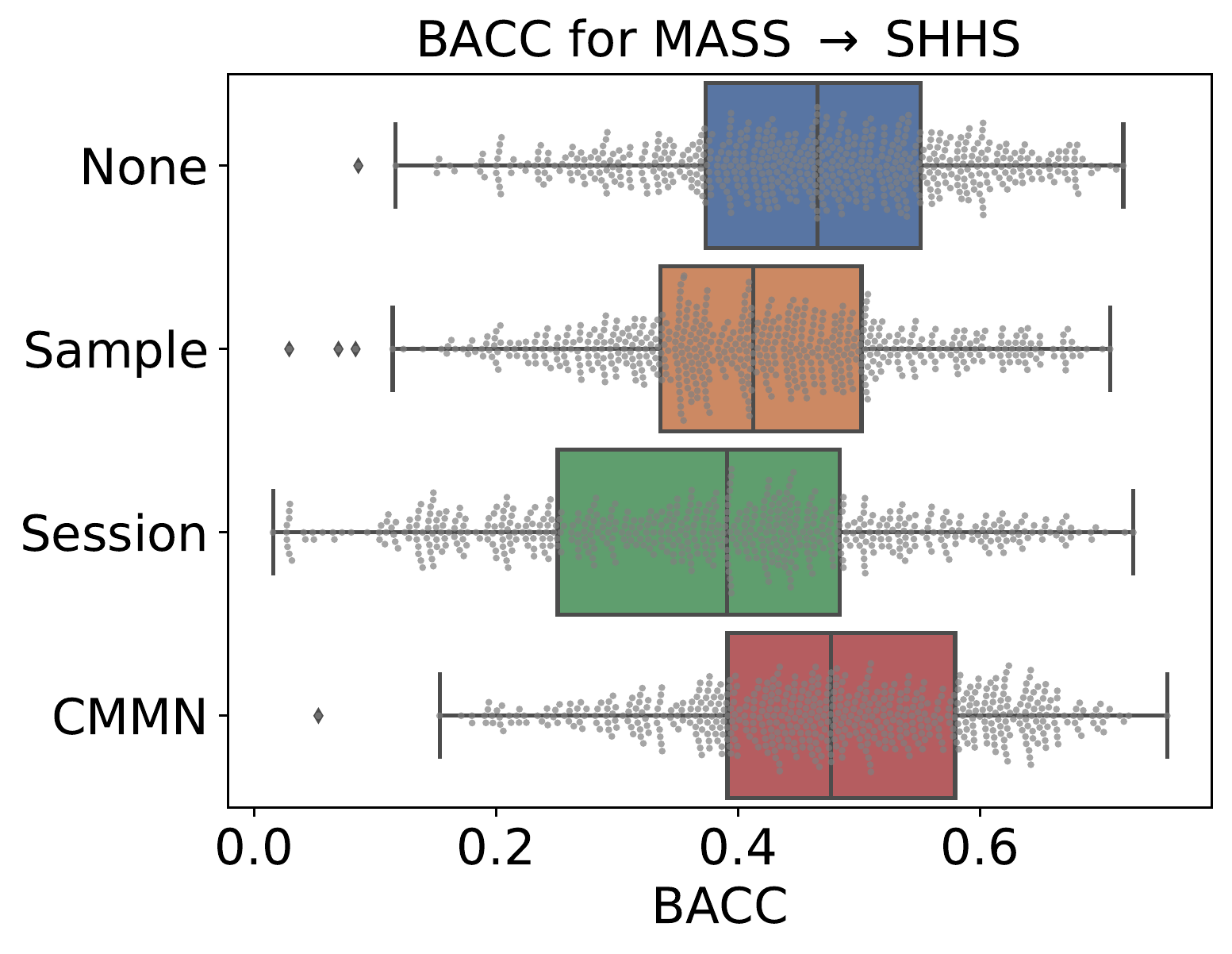}
  \caption{}
  \label{fig:sub2}
\end{subfigure}
\begin{subfigure}{.5\textwidth}
  \centering
  \includegraphics[width=.9\linewidth]{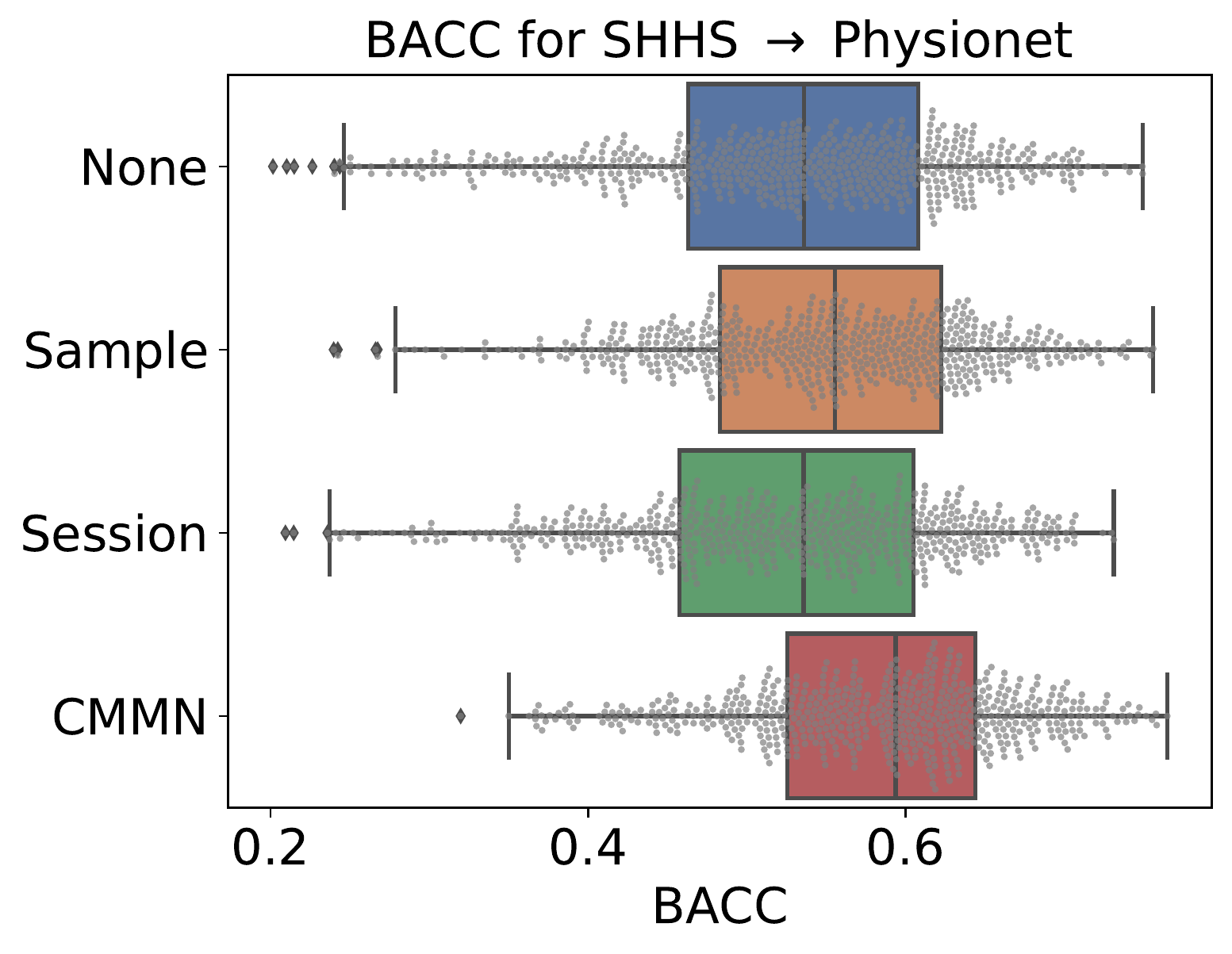}
  \caption{}
  \label{fig:sub1}
\end{subfigure}%
\begin{subfigure}{.5\textwidth}
  \centering
  \includegraphics[width=.9\linewidth]{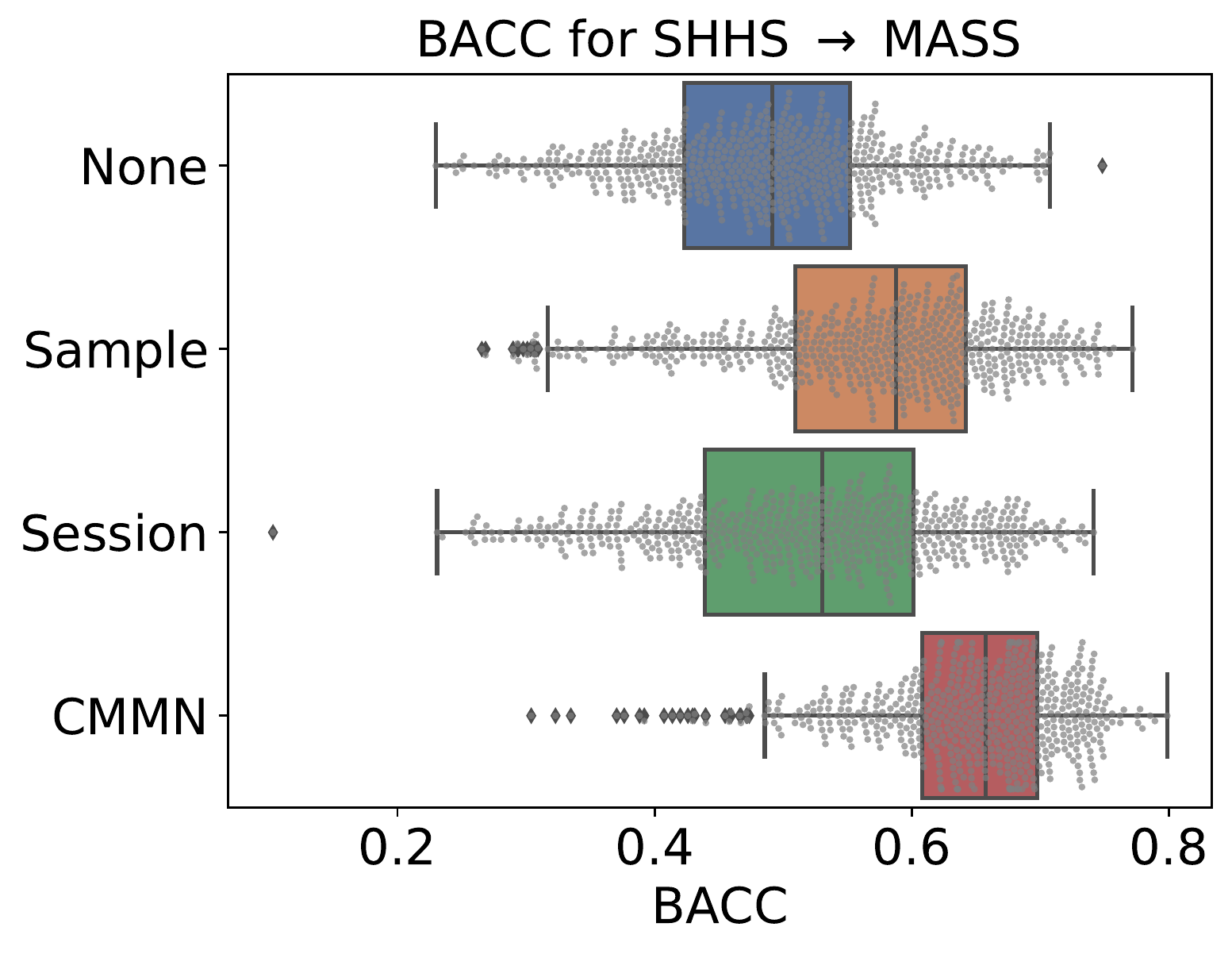}
  \caption{}
  \label{fig:sub2}
\end{subfigure}
\begin{subfigure}{.5\textwidth}
  \centering
  \includegraphics[width=.9\linewidth]{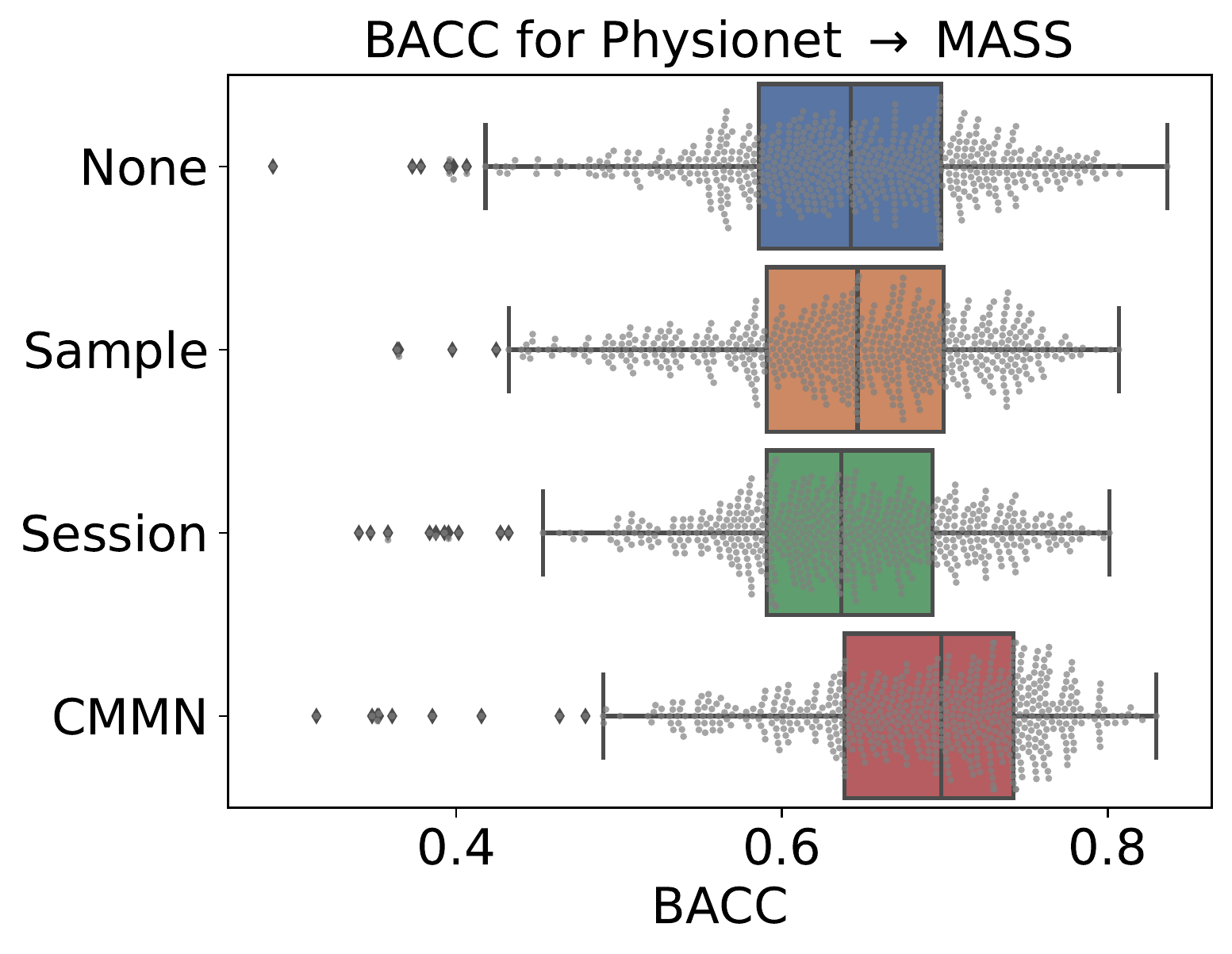}
  \caption{}
  \label{fig:sub1}
\end{subfigure}%
\begin{subfigure}{.5\textwidth}
  \centering
  \includegraphics[width=.9\linewidth]{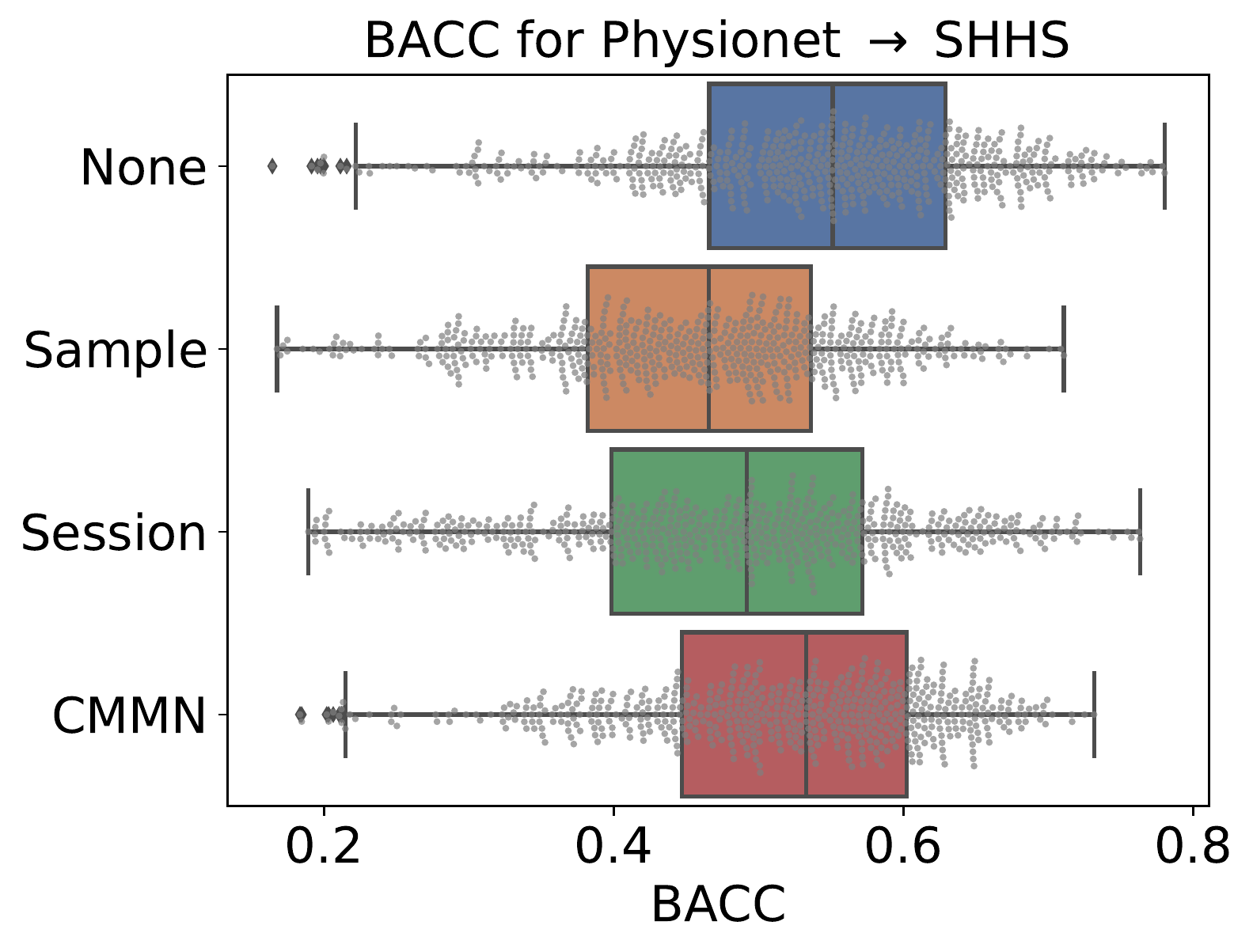}
  \caption{}
  \label{fig:boxplot_chambon}
\end{subfigure}
\caption{Boxplot of balanced accuracy (BACC)  for different normalizations 
and different train/test dataset pairs with \texttt{Chambon}.}
\end{figure}

\begin{figure}
\centering
\begin{subfigure}{.5\textwidth}
  \centering
  \includegraphics[width=.9\linewidth]{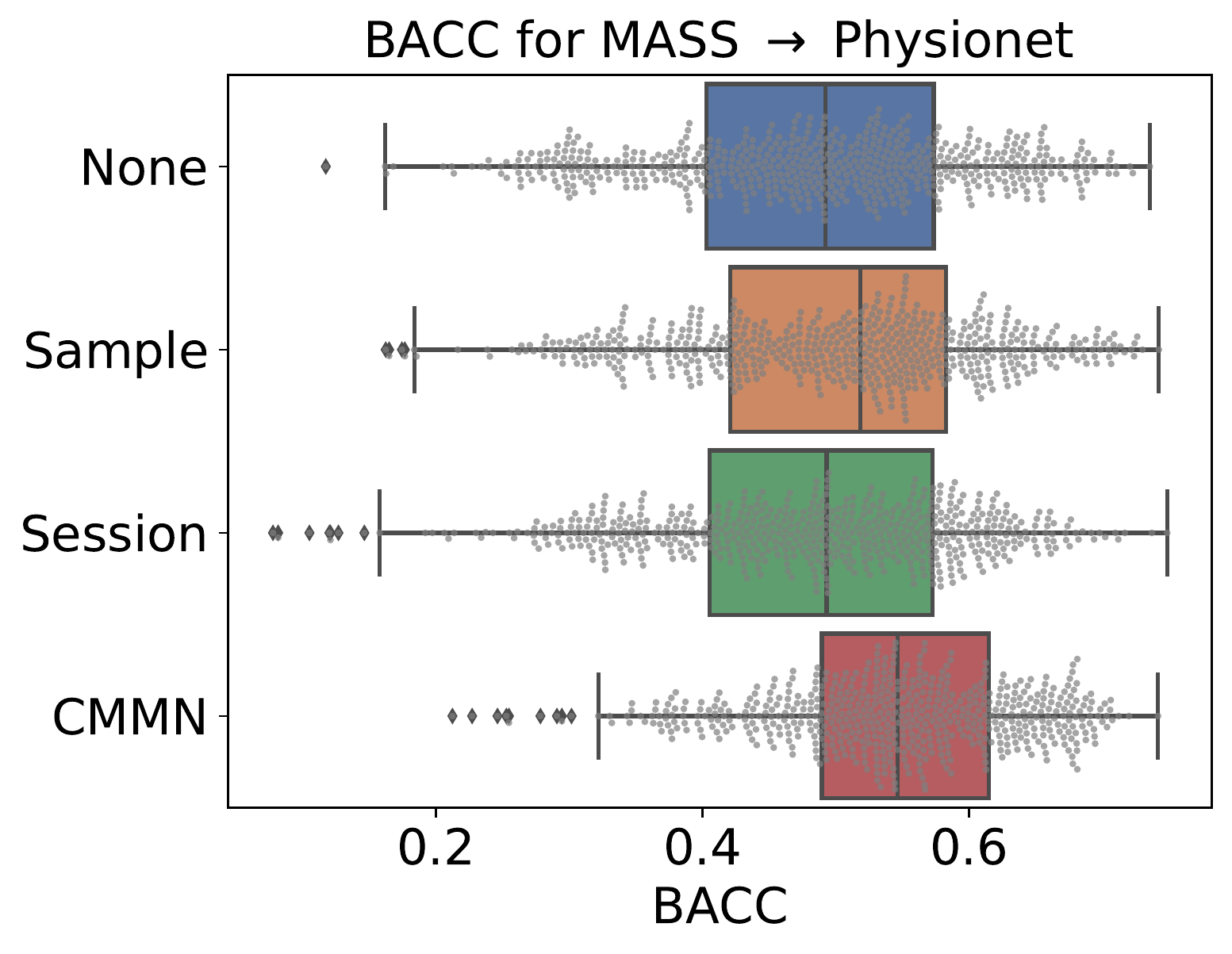}
  \caption{}
  \label{fig:sub1}
\end{subfigure}%
\begin{subfigure}{.5\textwidth}
  \centering
  \includegraphics[width=.9\linewidth]{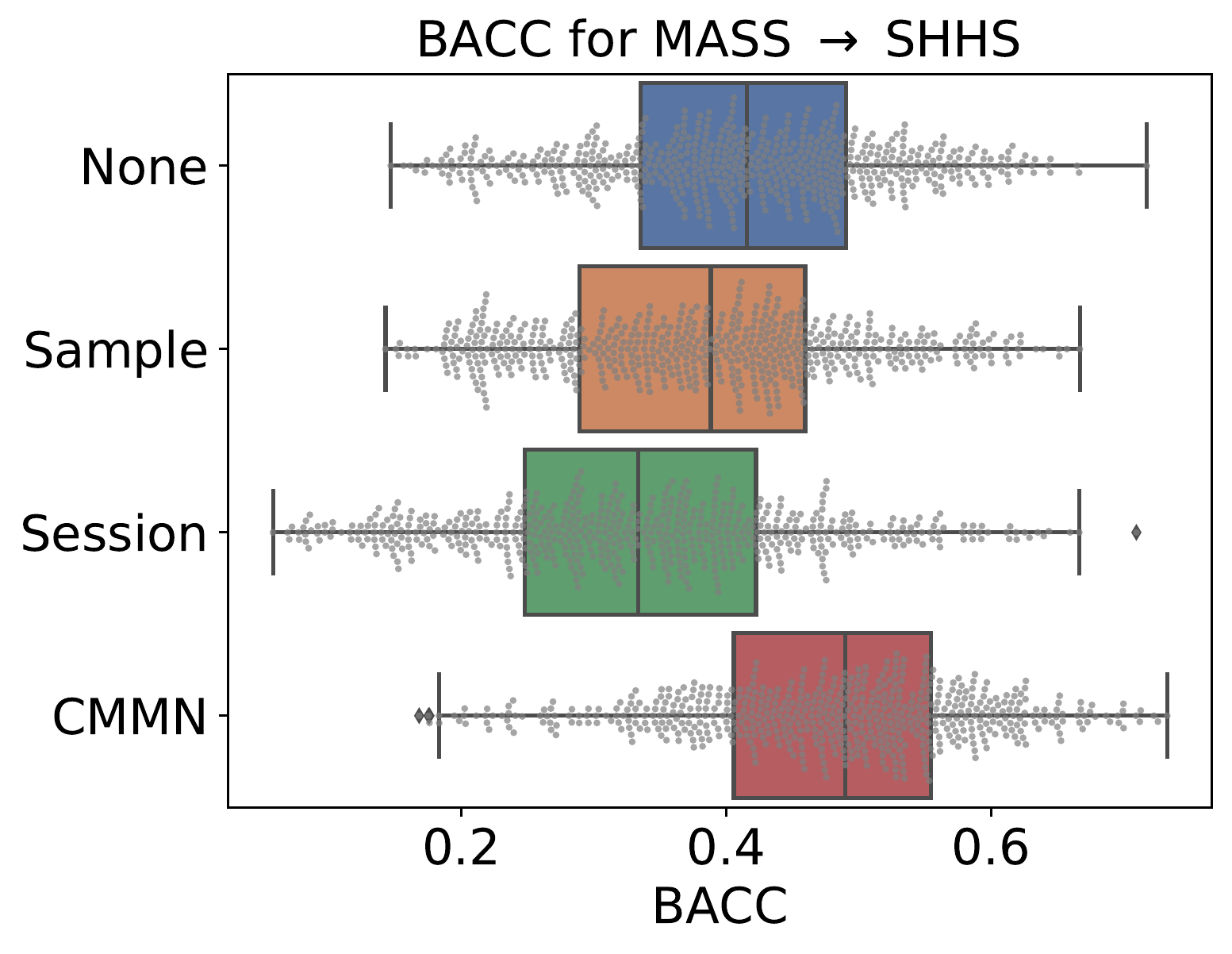}
  \caption{}
  \label{fig:sub2}
\end{subfigure}
\begin{subfigure}{.5\textwidth}
  \centering
  \includegraphics[width=.9\linewidth]{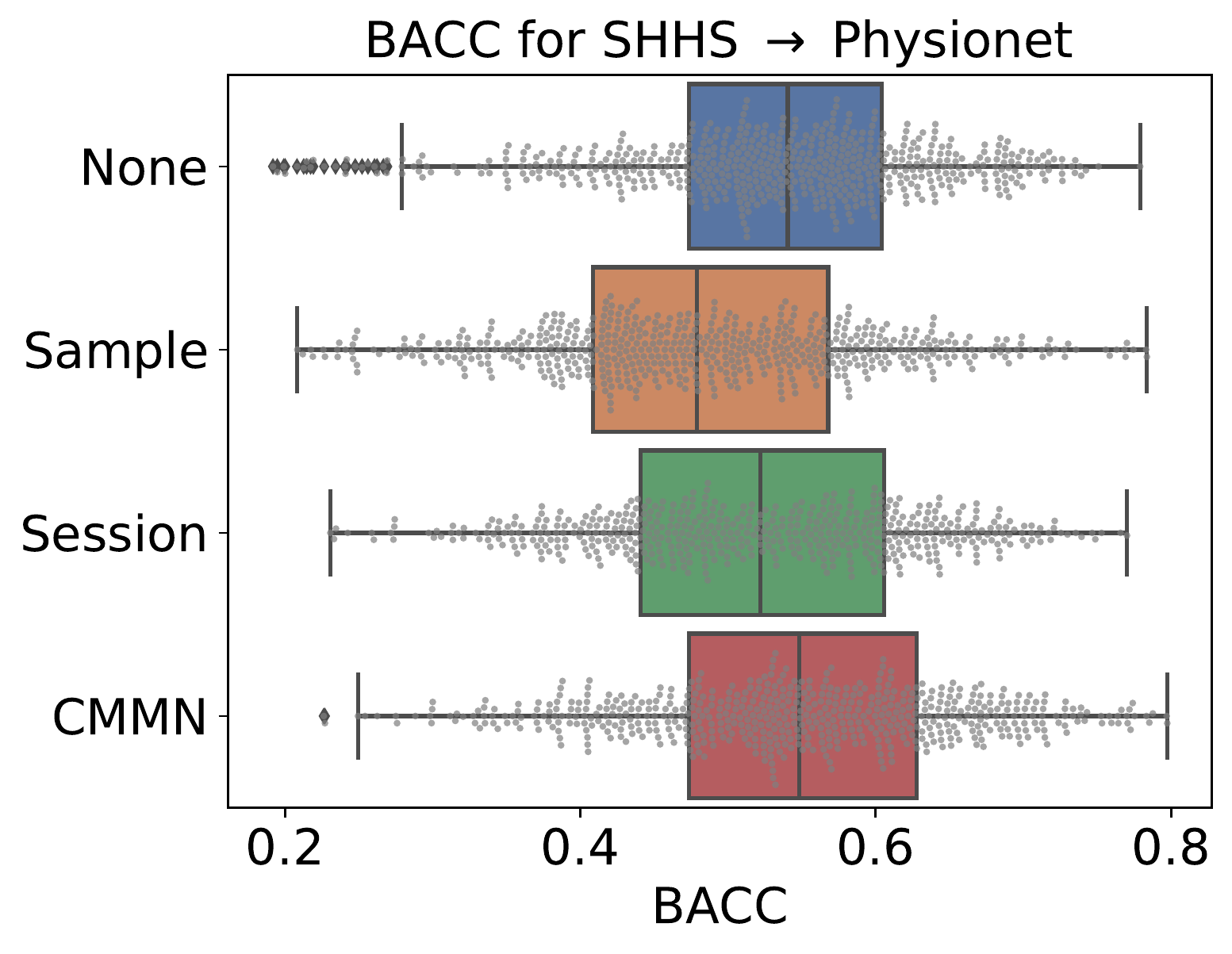}
  \caption{}
  \label{fig:sub1}
\end{subfigure}%
\begin{subfigure}{.5\textwidth}
  \centering
  \includegraphics[width=.9\linewidth]{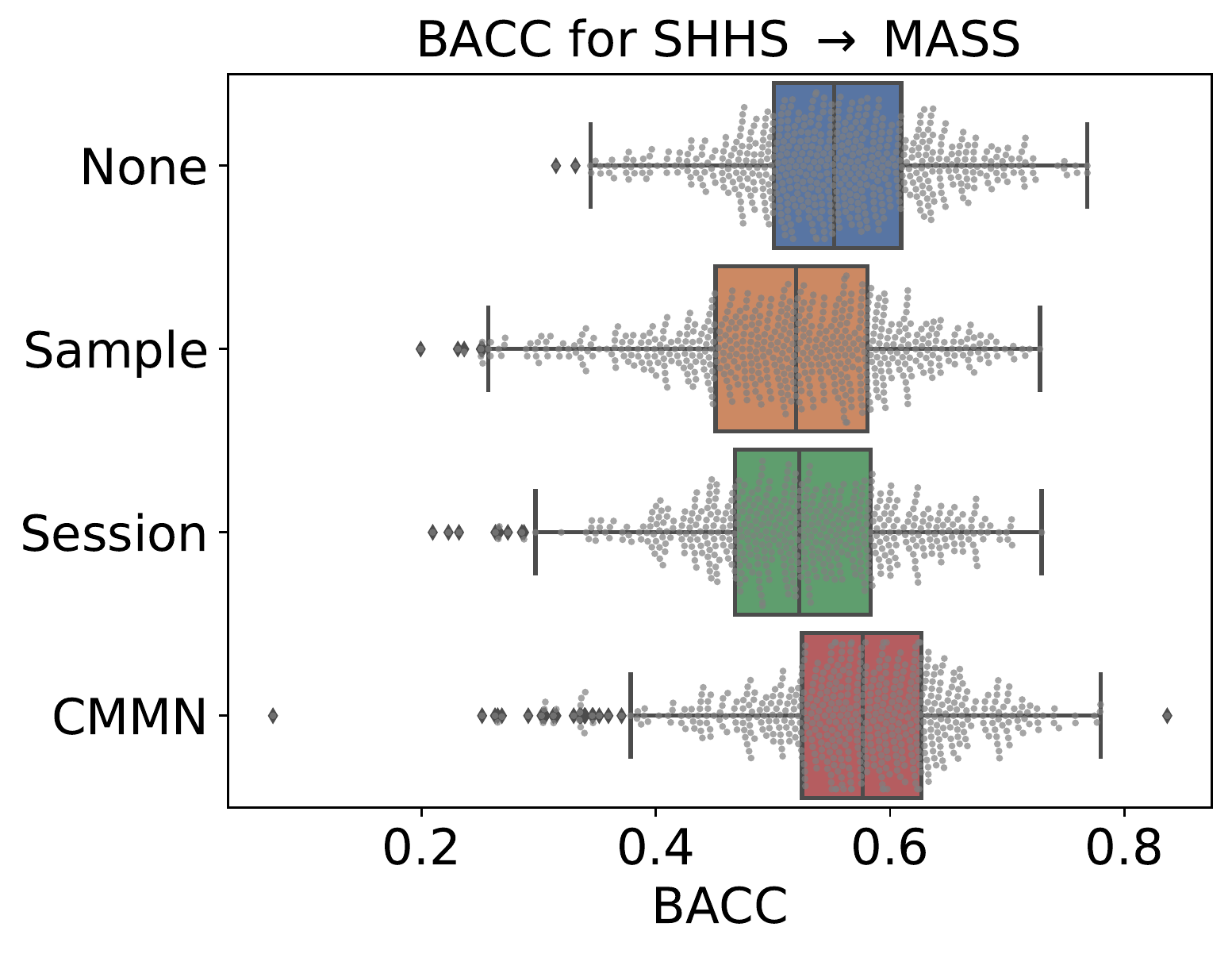}
  \caption{}
  \label{fig:sub2}
\end{subfigure}
\begin{subfigure}{.5\textwidth}
  \centering
  \includegraphics[width=.9\linewidth]{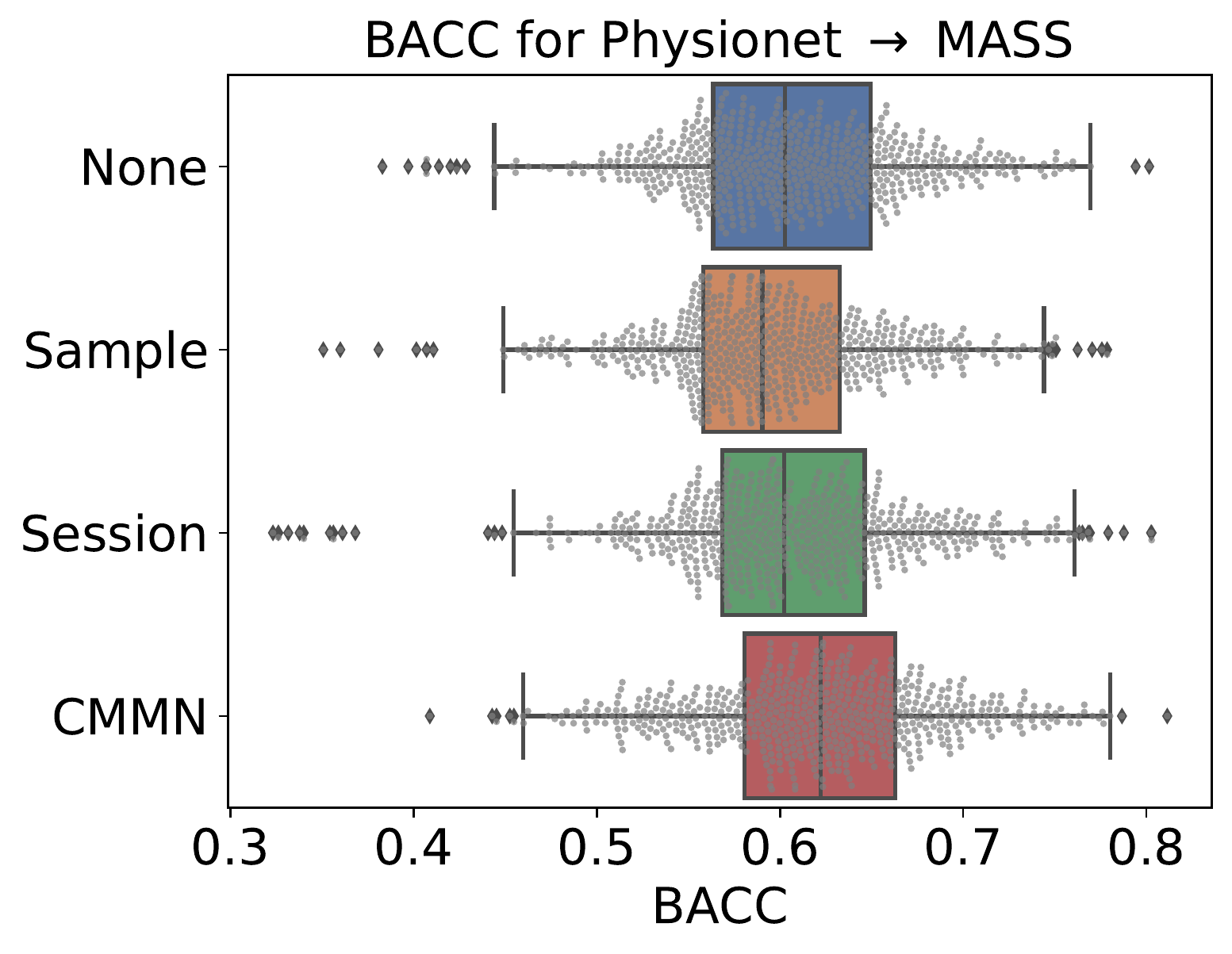}
  \caption{}
  \label{fig:sub1}
\end{subfigure}%
\begin{subfigure}{.5\textwidth}
  \centering
  \includegraphics[width=.9\linewidth]{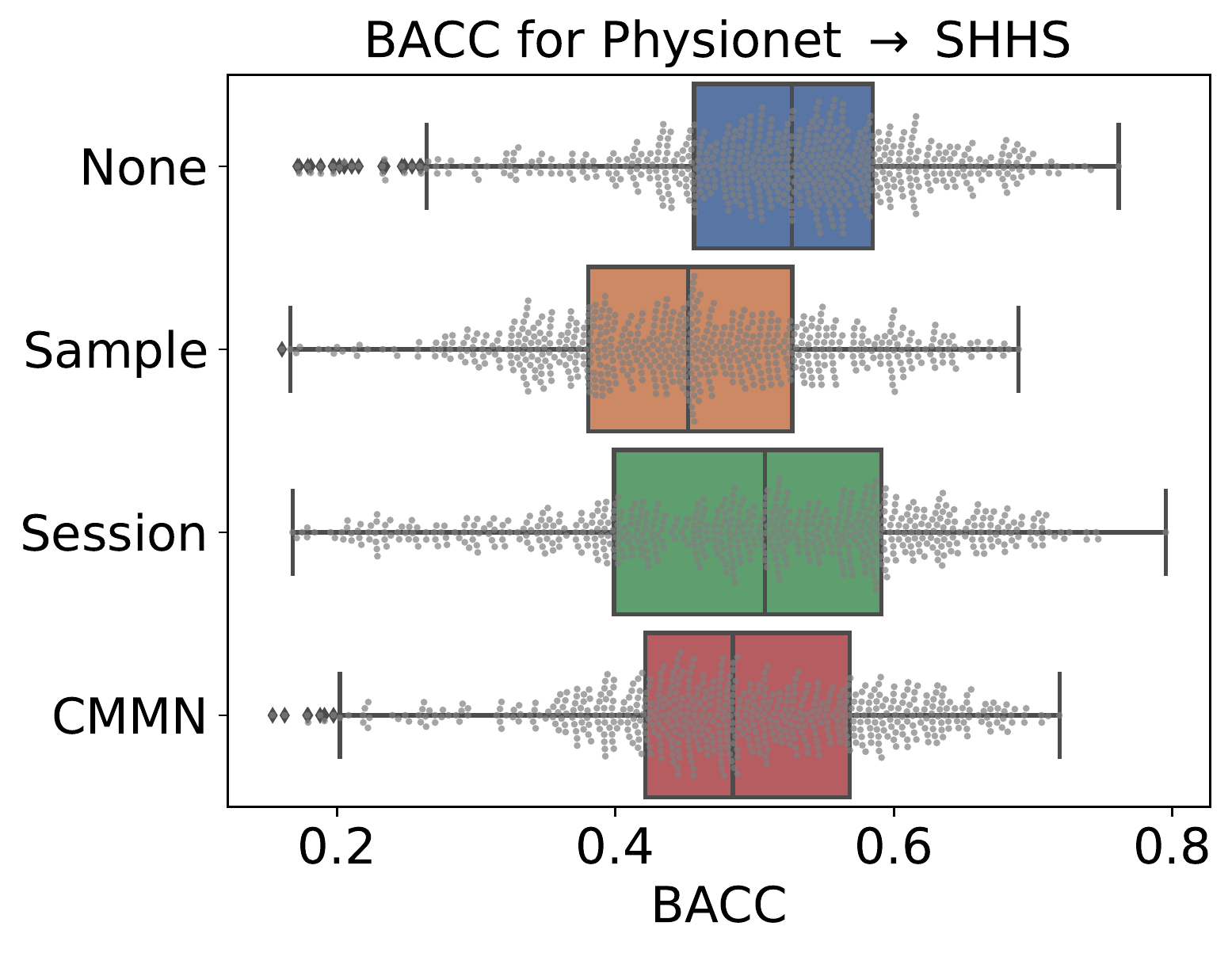}
  \caption{}
  \label{fig:boxplot_deep}
\end{subfigure}
\caption{Boxplot of balanced accuracy (BACC)  for different normalizations 
and different train/test dataset pairs with \texttt{DeepSleepNet}.}
\end{figure}

\begin{figure}
\begin{subfigure}{.5\textwidth}
  \centering
  \includegraphics[width=.9\linewidth]{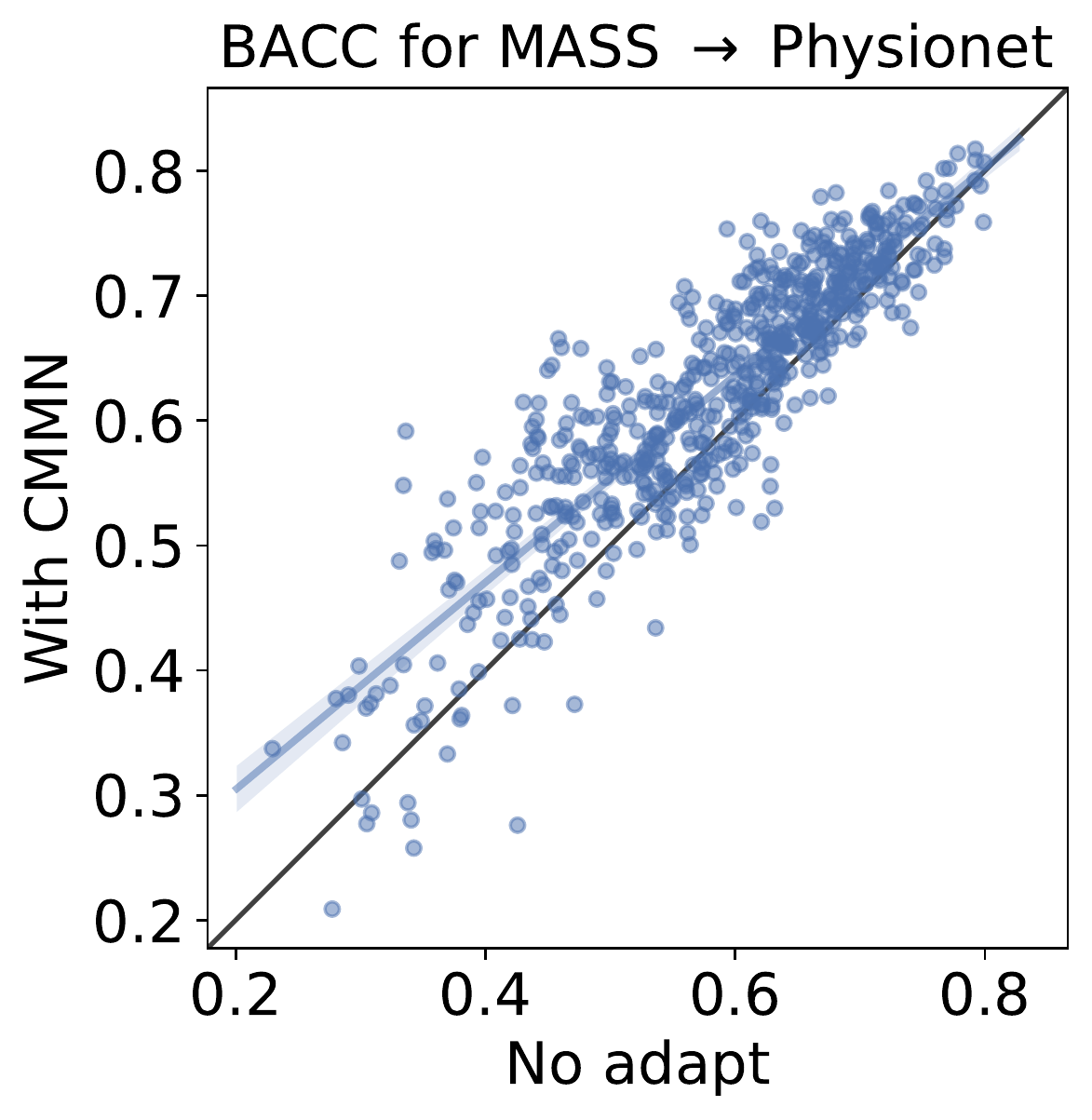}
  \caption{}
  \label{fig:sub1}
\end{subfigure}%
\begin{subfigure}{.5\textwidth}
  \centering
  \includegraphics[width=.9\linewidth]{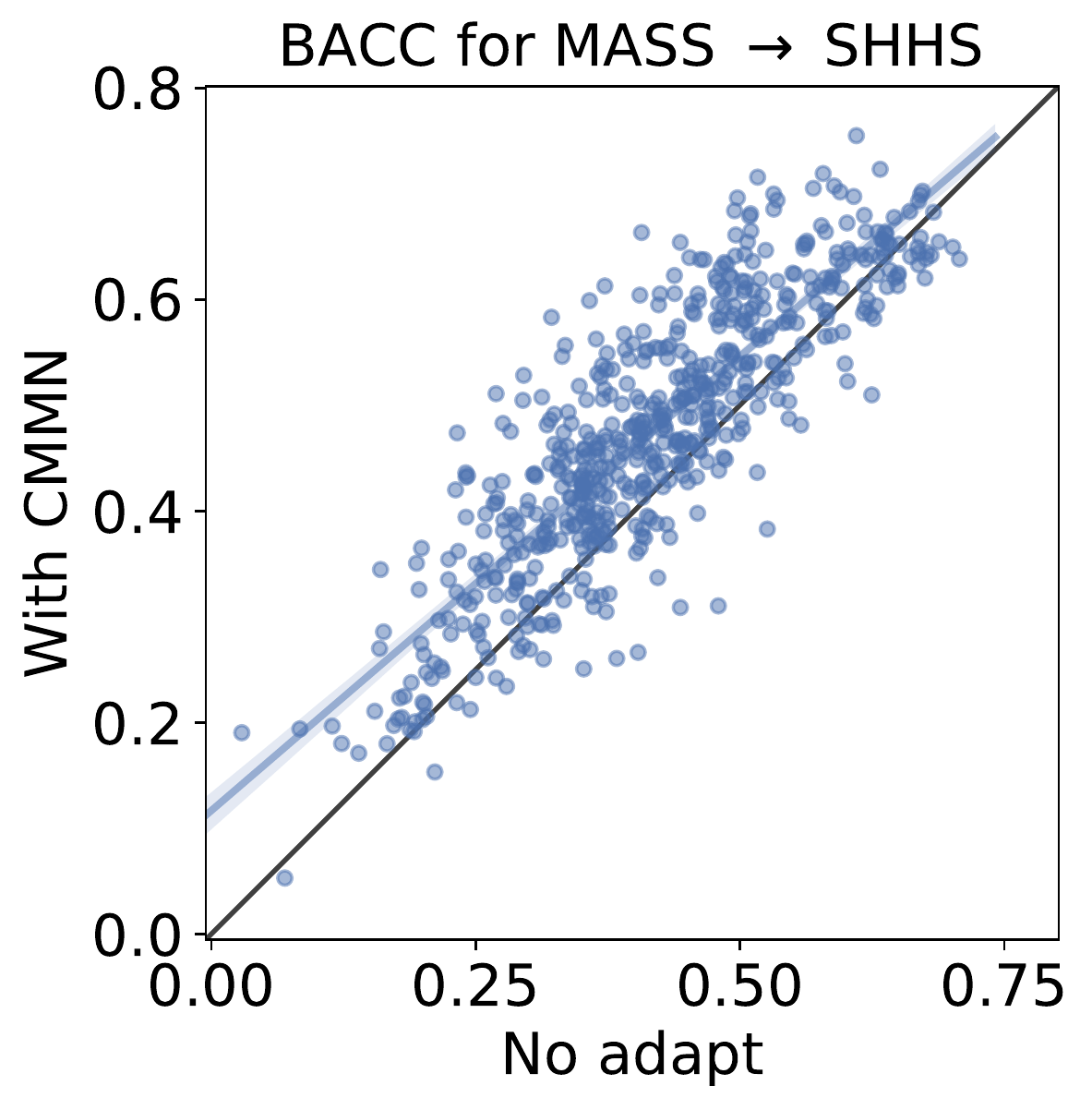}
  \caption{}
  \label{fig:sub2}
\end{subfigure}
\begin{subfigure}{.5\textwidth}
  \centering
  \includegraphics[width=.9\linewidth]{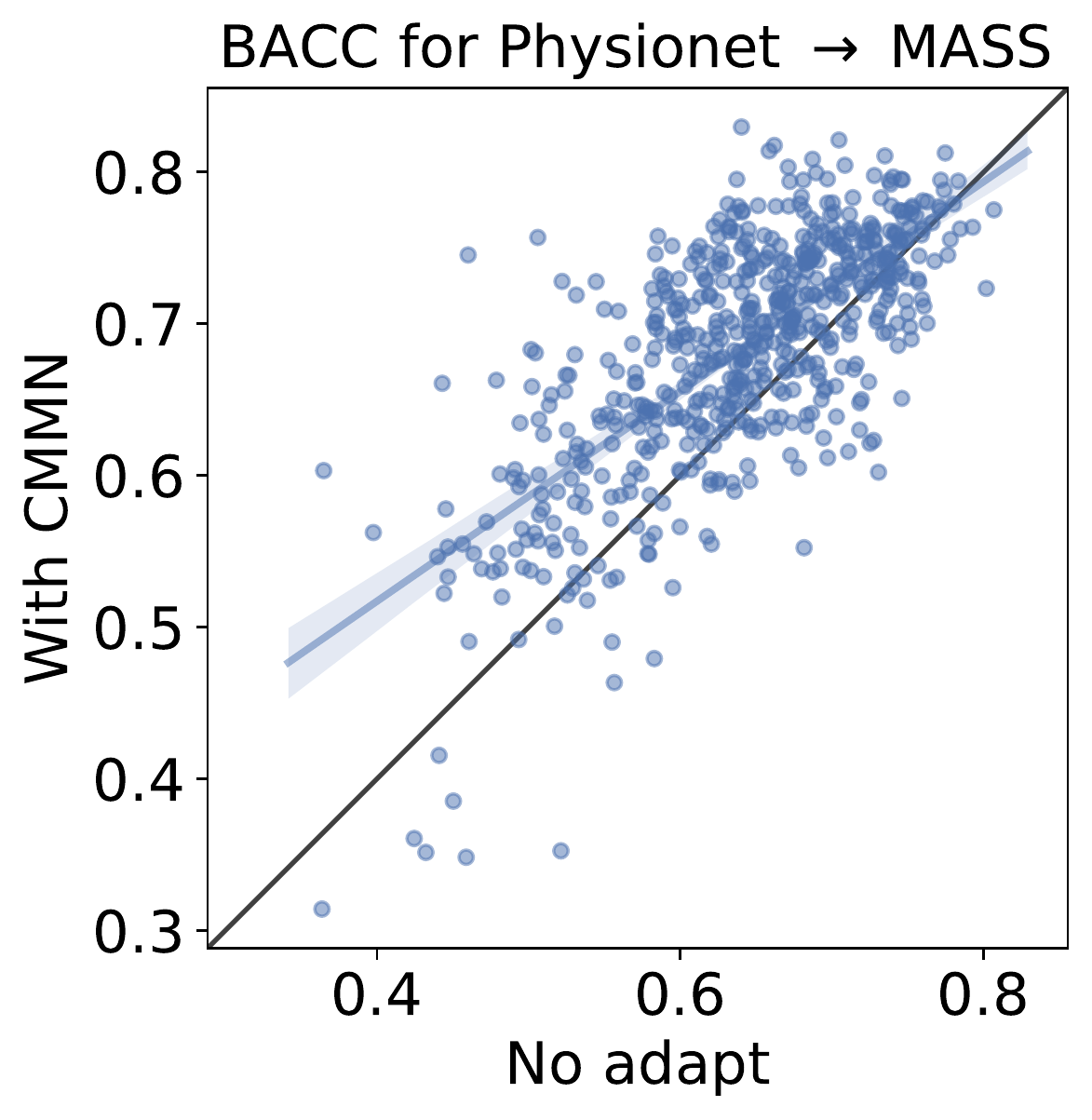}
  \caption{}
  \label{fig:sub1}
\end{subfigure}%
\begin{subfigure}{.5\textwidth}
  \centering
  \includegraphics[width=.9\linewidth]{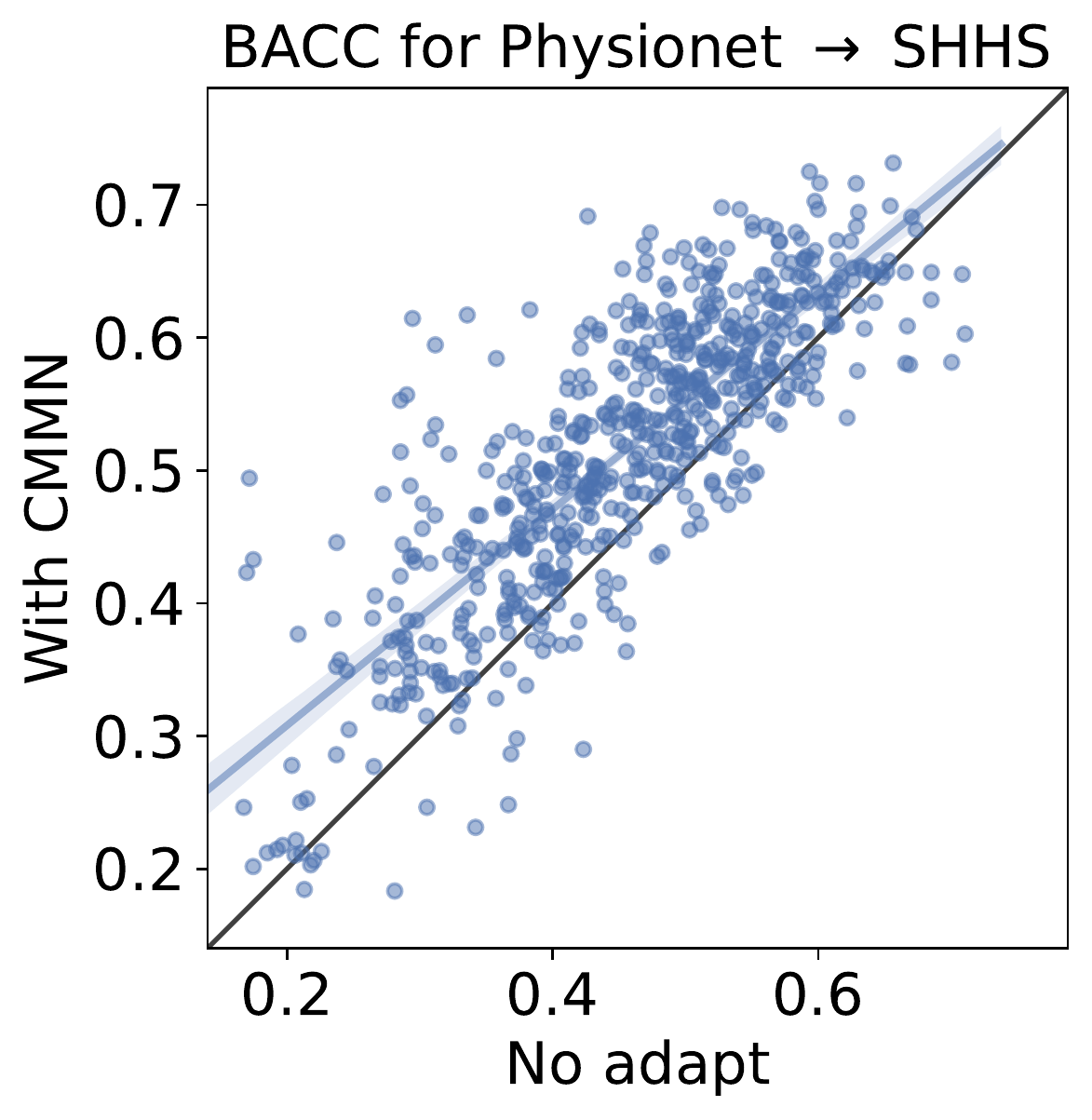}
  \caption{}
  \label{fig:sub2}
\end{subfigure}
\begin{subfigure}{.5\textwidth}
  \centering
  \includegraphics[width=.9\linewidth]{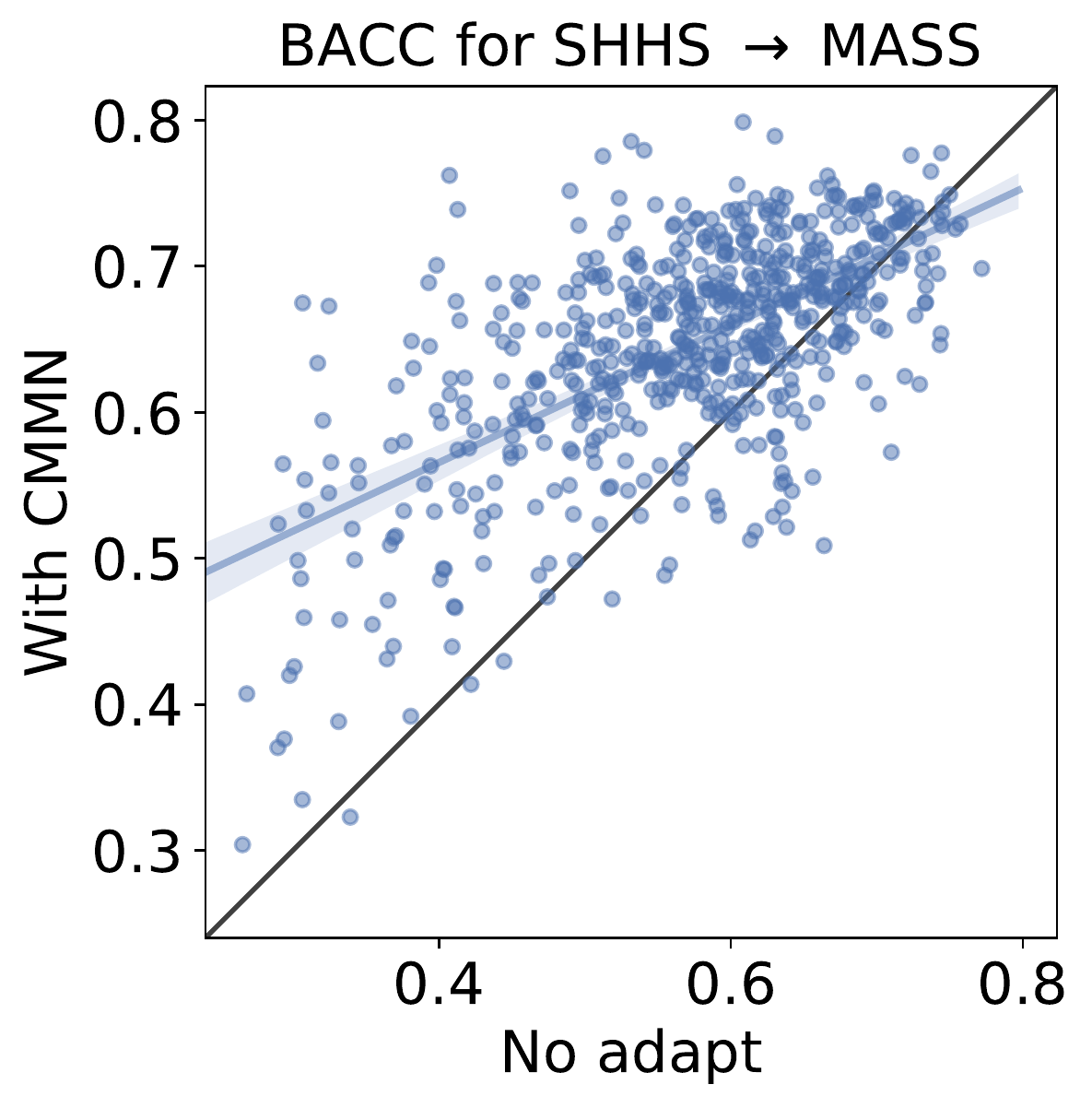}
  \caption{}
  \label{fig:sub1}
\end{subfigure}%
\begin{subfigure}{.5\textwidth}
  \centering
  \includegraphics[width=.9\linewidth]{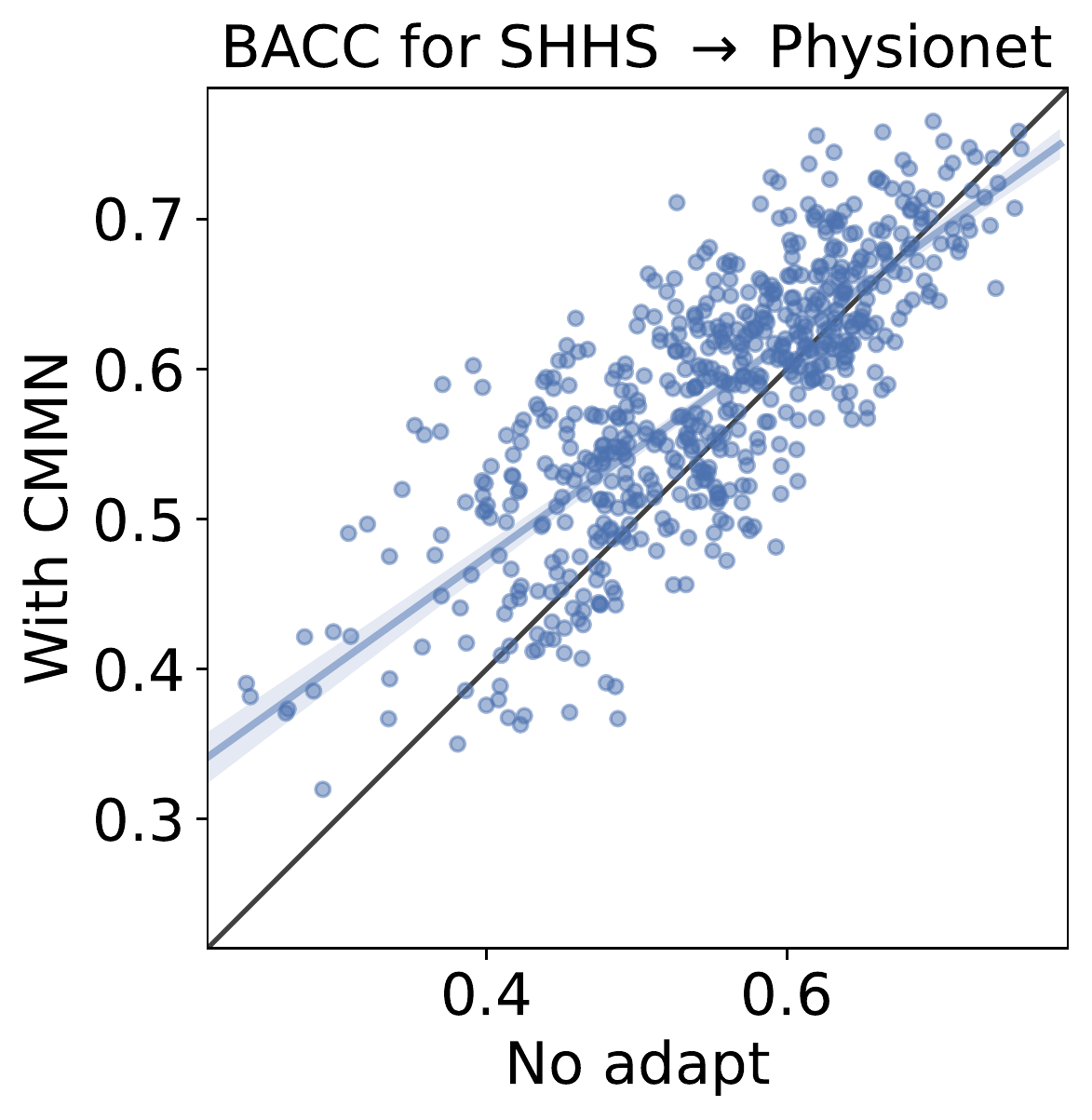}
  \caption{}
  \label{fig:scatter_chambon}
\end{subfigure}
\caption{Scatter plot of balanced accuracy (BACC) with \texttt{No Adapt} as a function
    of BACC with \texttt{CMMN} for different dataset pairs
    with \texttt{Chambon}.}
\label{fig:test}
\end{figure}

\subsection{Scatterplot for different architectures and different dataset pairs}

One significant benefit of \texttt{CMMN} is to have a massive impact on low-performing subjects. The scatter plots of balanced accuracy (BACC) with \texttt{No Adapt} as a function of BACC with \texttt{CMMN} emphasize this effect. In this section, the plots in figure \ref{fig:scatter_chambon} and figure \ref{fig:scatter_deep} show the increase in all subjects. For both architectures, the axis of the linear regression is always above the axis $x=y$. The differences between the two axes are generally higher for lower performances, which means that the worst the performance is, the higher the increase is. For the eccentric dots on the left (\ie low-performing subjects), the increase is generally around 8\% (see \autoref{tab:delta}). When SHHS is the training set, low-performing subjects get a considerable boost. On the other hand, for Physionet $\leftrightarrow$ MASS with \texttt{Chambon}, some subjects lose accuracy, but it remains rare compared to the average gain on all subjects. 
\begin{figure}
\begin{subfigure}{.5\textwidth}
  \centering
  \includegraphics[width=.9\linewidth]{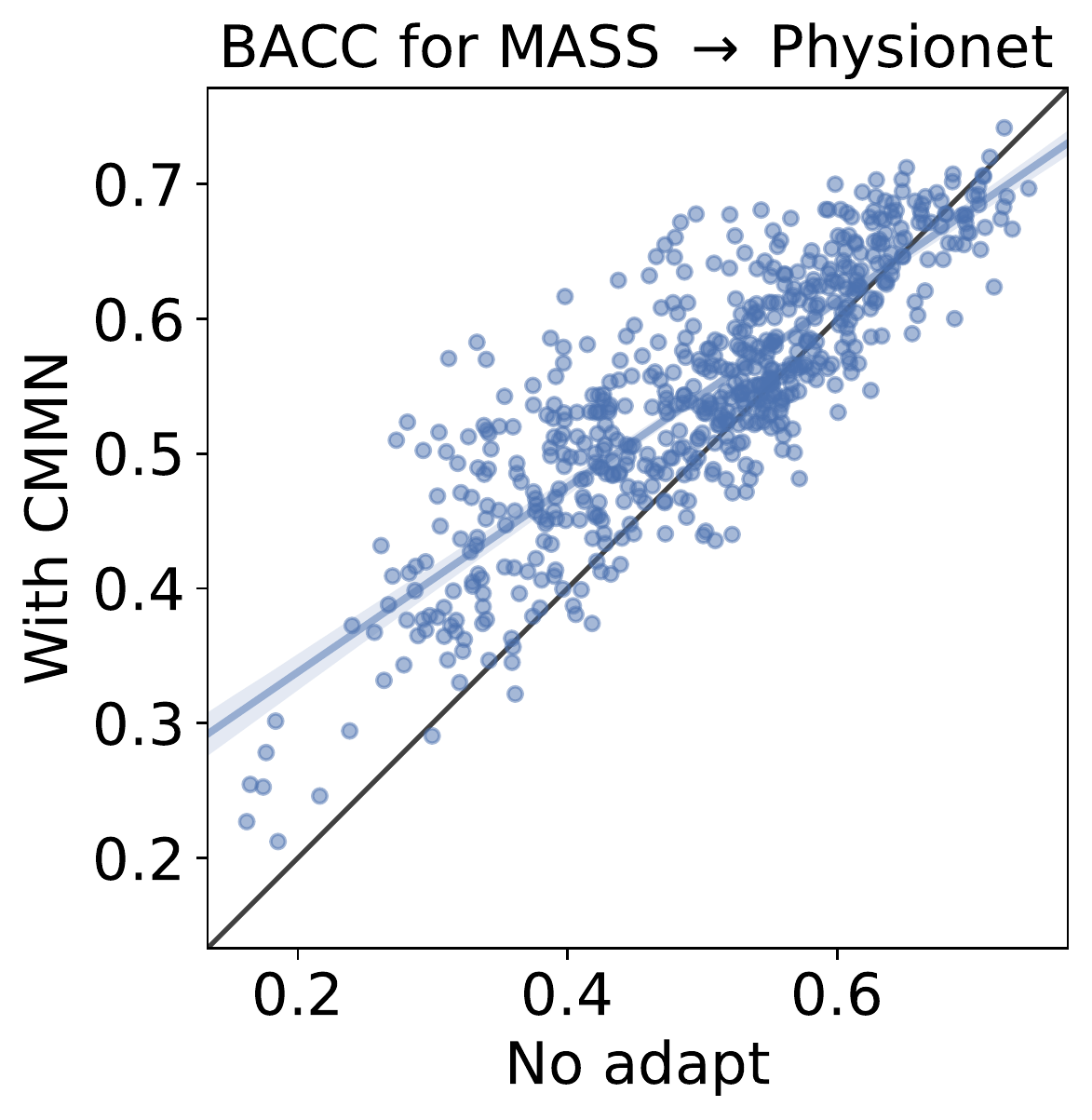}
  \caption{}
  \label{fig:sub1}
\end{subfigure}%
\begin{subfigure}{.5\textwidth}
  \centering
  \includegraphics[width=.9\linewidth]{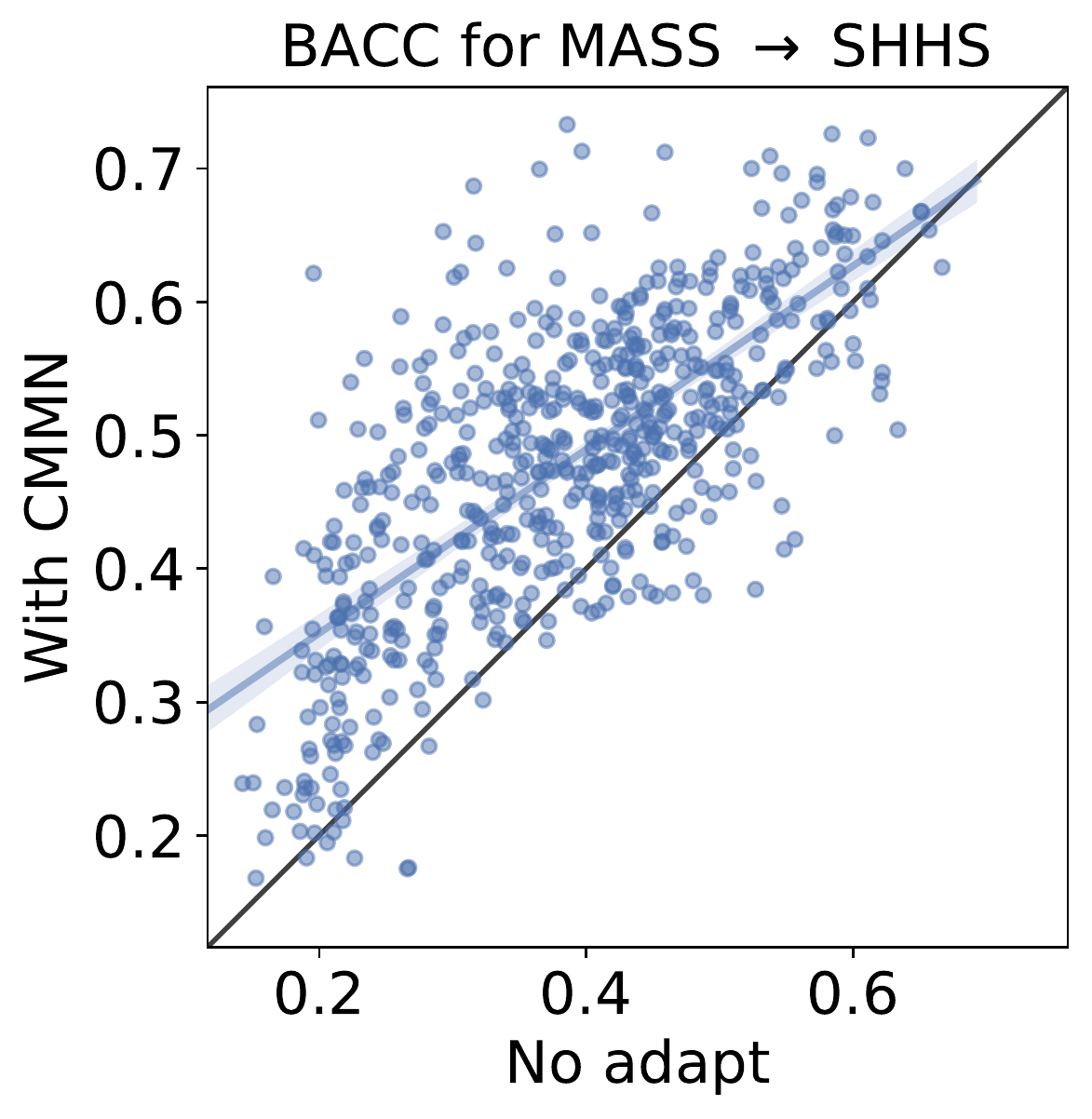}
  \caption{}
  \label{fig:sub2}
\end{subfigure}
\begin{subfigure}{.5\textwidth}
  \centering
  \includegraphics[width=.9\linewidth]{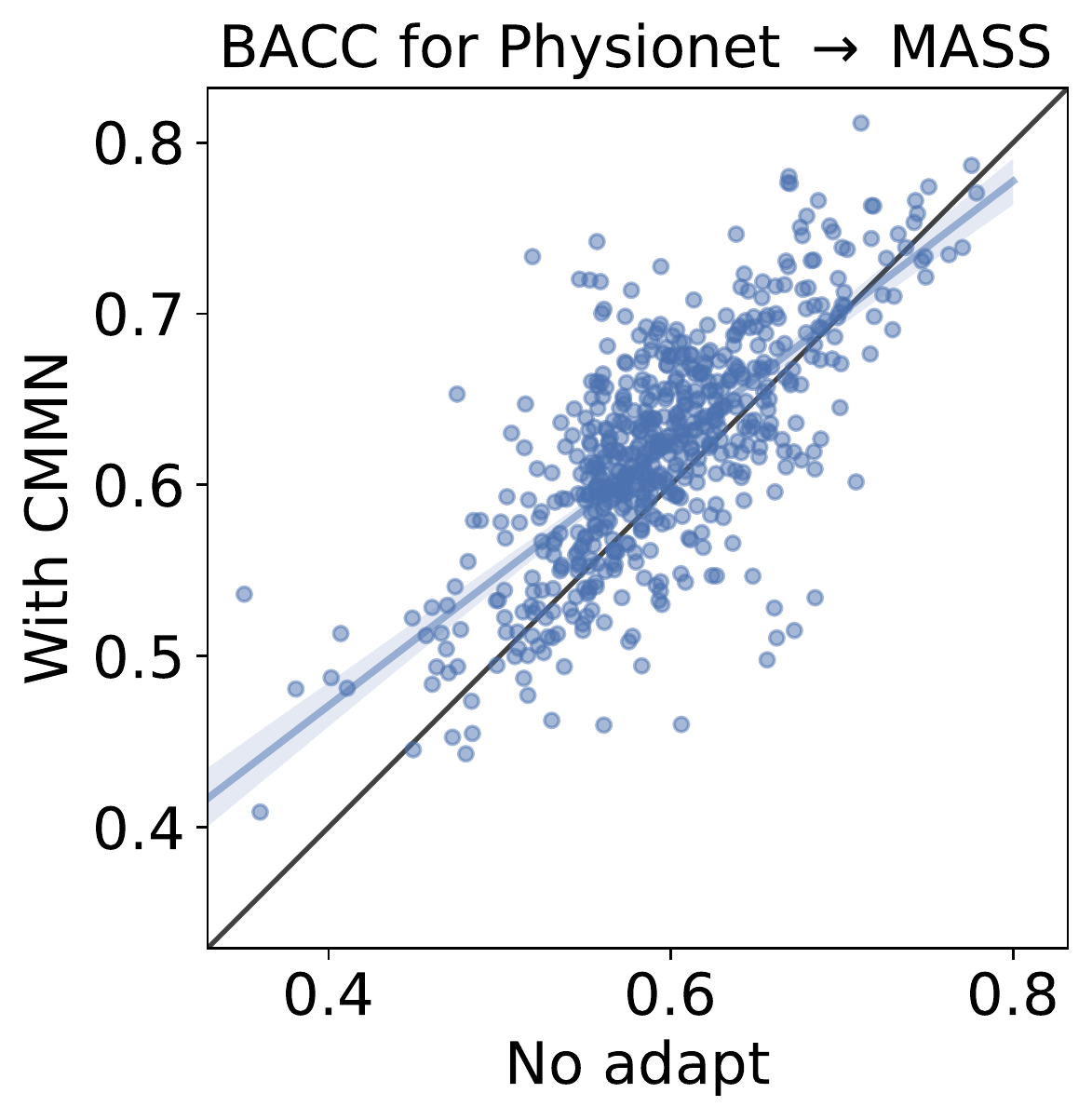}
  \caption{}
  \label{fig:sub1}
\end{subfigure}%
\begin{subfigure}{.5\textwidth}
  \centering
  \includegraphics[width=.9\linewidth]{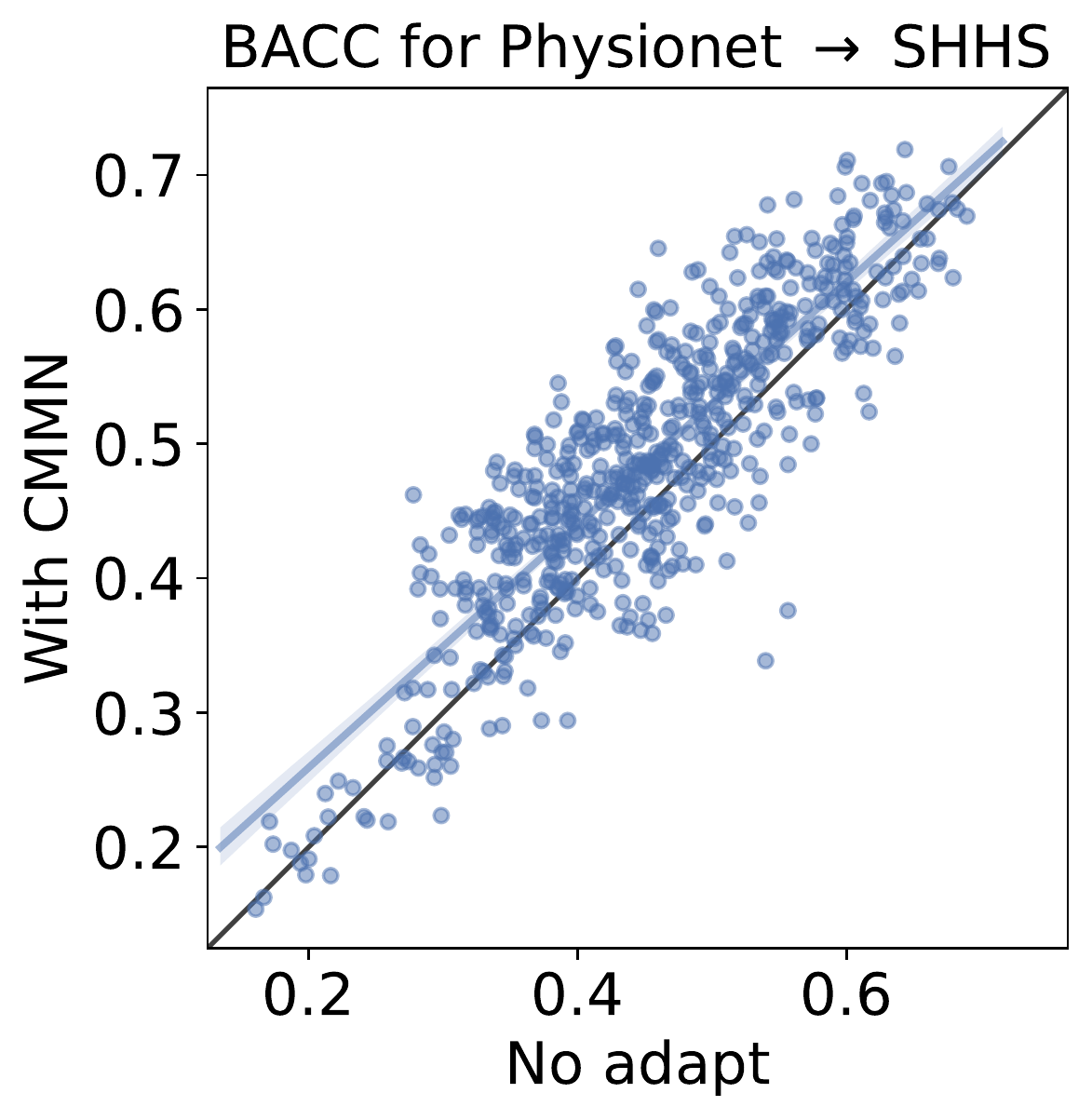}
  \caption{}
  \label{fig:sub2}
\end{subfigure}
\begin{subfigure}{.5\textwidth}
  \centering
  \includegraphics[width=.9\linewidth]{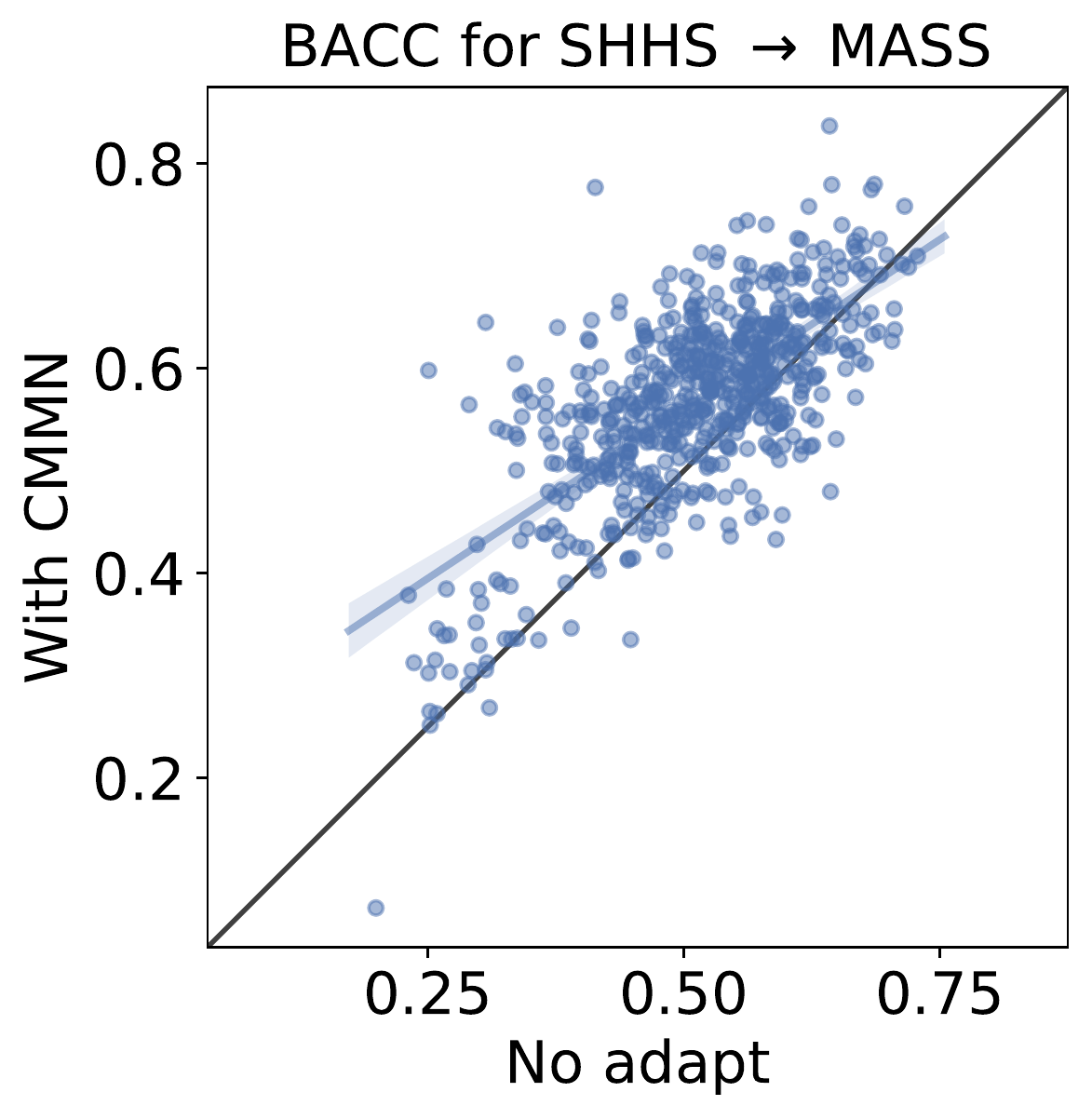}
  \caption{}
  \label{fig:sub1}
\end{subfigure}%
\begin{subfigure}{.5\textwidth}
  \centering
  \includegraphics[width=.9\linewidth]{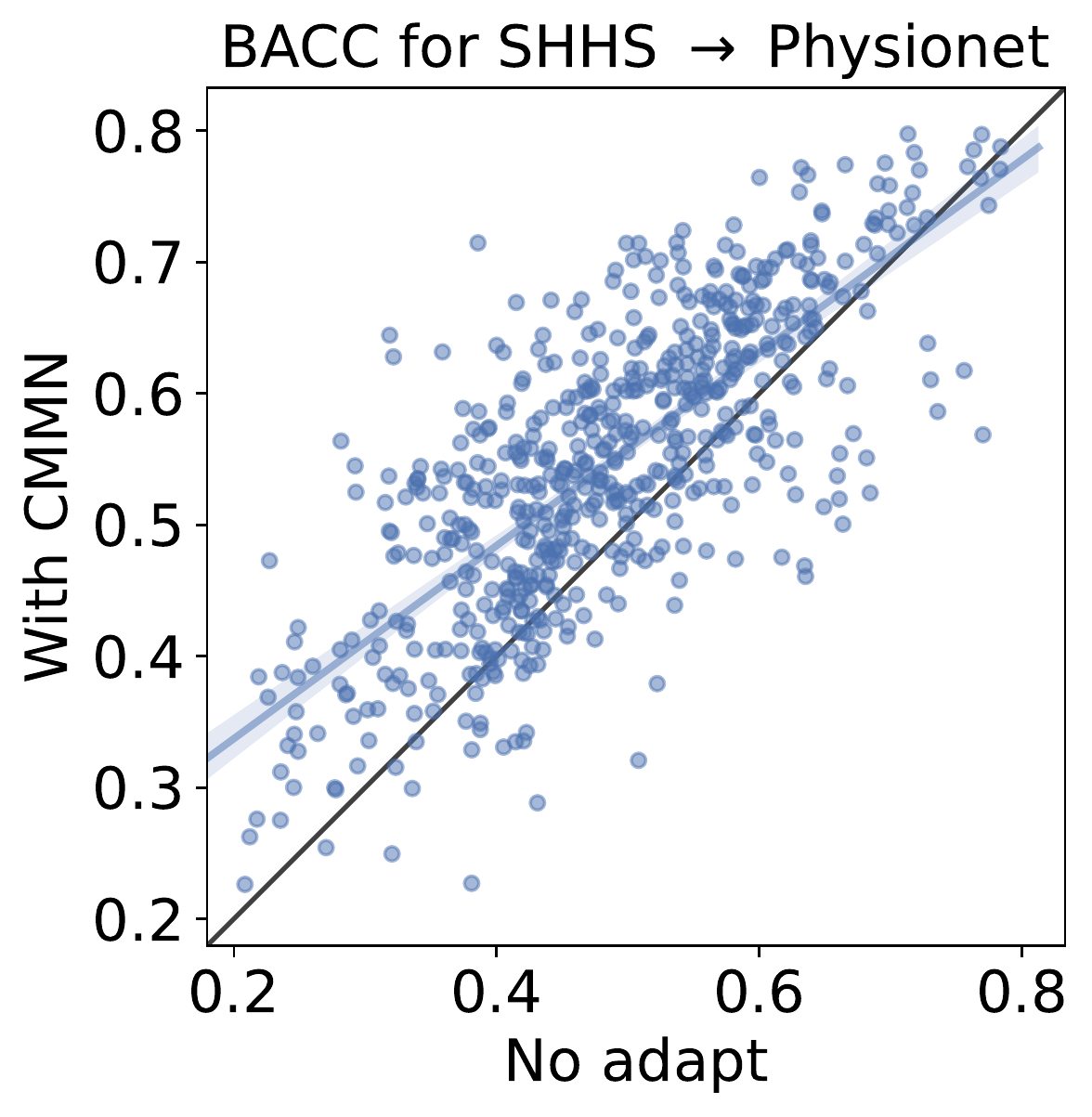}
  \caption{}
  \label{fig:scatter_deep}
\end{subfigure}
\caption{Scatter plot of balanced accuracy (BACC) with \texttt{No Adapt} as a function
    of BACC with \texttt{CMMN} for different dataset pairs
    with \texttt{DeepSleepNet}}
\label{fig:test}
\end{figure}


\end{document}